\documentclass[onecolumn,11pt,notitlepage,nofootinbib,floatfix,superscriptaddress]{revtex4-2}
\usepackage{amsmath, amsfonts, amsthm, amssymb, bbm}
\usepackage[margin=1in]{geometry}
\usepackage[parfill]{parskip}
\usepackage{algorithm,algpseudocode}
\usepackage{graphicx,xcolor,tikz}
\usepackage{hyperref}
\usepackage{mathtools}
\mathtoolsset{showonlyrefs}

\makeatletter
\renewcommand{\p@subsection}{}
\makeatother

\newcommand*{\rom}[1]{\expandafter\@slowromancap\romannumeral #1@}

\newtheorem{theorem}{Theorem}[section]
\newtheorem{lemma}[theorem]{Lemma}
\newtheorem{definition}[theorem]{Definition}

\newtheorem{corollary}[theorem]{Corollary}

\newcommand{\al}{\alpha}
\newcommand{\be}{\beta}

\newcommand{\ep}{\varepsilon}
\newcommand{\td}{\tilde}
\newcommand{\bb}{\mathbb}
\newcommand{\mc}{\mathcal}
\newcommand{\Cay}{\operatorname{Cay}}

\newcommand{\spn}{\operatorname{span}}

\algblockdefx[RepeatN]{RepeatN}{EndRepeatN}[1]{\textbf{repeat} #1 \textbf{times}}{\textbf{end repeat}}

\algblockdefx[ParallelForEach]{ParallelForEach}{EndParallelForEach}[1]{\textbf{parallel for each} #1 \textbf{do}}{\textbf{end parallel for each}}

\algnewcommand{\LongState}[1]{\State%
  \parbox[t]{\dimexpr\linewidth -\algorithmicindent}{\strut #1\strut}}

\begin{document}
	
\title{Single-shot decoding of good quantum LDPC codes}

\author{Shouzhen Gu}
\thanks{Institute for Quantum Information and Matter, California Institute of Technology, Pasadena, CA, USA}
\author{Eugene Tang}
\thanks{Center for Theoretical Physics, Massachusetts Institute of Technology, Cambridge, MA, USA}
\author{Libor Caha}
\thanks{School of Computation, Information and Technology, Technical University of Munich \&
Munich Center for Quantum Science and Technology, Munich, Germany}
\author{\\Shin Ho Choe}
\thanks{School of Computation, Information and Technology, Technical University of Munich \&
Munich Center for Quantum Science and Technology, Munich, Germany}
\author{Zhiyang He (Sunny)}
\thanks{Department of Mathematics, Massachusetts Institute of Technology, Cambridge, MA, USA}
\author{Aleksander Kubica}
\thanks{AWS Center for Quantum Computing \& California Institute of Technology, Pasadena, CA, USA}

\begin{abstract}
Quantum Tanner codes constitute a family of quantum low-density parity-check (LDPC) codes with good parameters, i.e., constant encoding rate and relative distance. In this article, we prove that quantum Tanner codes also facilitate single-shot quantum error correction (QEC) of adversarial noise, where one measurement round (consisting of constant-weight parity checks) suffices to perform reliable QEC even in the presence of measurement errors. We establish this result for both the sequential and parallel decoding algorithms introduced by Leverrier and Z\'emor. Furthermore, we show that in order to suppress errors over multiple repeated rounds of QEC, it suffices to run the parallel decoding algorithm for constant time in each round. Combined with good code parameters, the resulting constant-time overhead of QEC and robustness to (possibly time-correlated) adversarial noise make quantum Tanner codes
alluring from the perspective of quantum fault-tolerant protocols. 
\end{abstract}

\maketitle

\section{Introduction}
Quantum error correcting (QEC) codes~\cite{Shor1995,Steane1996} are the backbone of quantum fault-tolerant protocols needed to reliably operate scalable quantum computers.
Due to their simplicity, stabilizer codes~\cite{Gottesman1996}, which can be realized by measuring a set of commuting Pauli operators known as parity checks, have received much attention.
From the perspective of fault tolerance, it might be desirable to further require that qubits are placed on some lattice and to restrict parity checks to be constant-weight and geometrically local.
However, such topological QEC codes, which include the toric code~\cite{Kitaev1997,dennisTopologicalQuantumMemory2002} and the color code~\cite{Bombin2006,Bombin2007,Kubicathesis} as examples, have limited code parameters~\cite{BravyiTerhal2009,BravyiPoulinTerhal2010,Baspin2022param}. 
To avoid these limitations, one can drop the assumption about geometric locality of parity checks (while still maintaining the assumption about their constant weight) to obtain a more general family of QEC codes known as quantum low-density parity-check (QLDPC) codes; see Ref.~\cite{breuckmann2021} for a recent review.
Importantly, QLDPC codes can have essentially optimal parameters, as shown by recent breakthrough results~\cite{evra2020decodable,hastings2020fiber,PKAlmostLinear,Breuckmann2020}, culminating in the construction of (asymptotically) good QLDPC codes whose encoding rates and relative distances are constant~\cite{PK21}.
A key component of the construction of asymptotically good QLDPC codes is the presence of ``product-expanding'' local codes. Since then, a few alternative constructions of good QLDPC codes have been proposed~\cite{quantumTannerCodes, dinur2022good}.

Good parameters alone are not enough for QEC codes to be interesting beyond the theoretical realm. In order to be practically relevant and useful, QEC codes need computationally efficient decoding algorithms which process the error syndrome and identify errors afflicting the encoded information. Importantly, decoding algorithms need to operate at least at the speed at which quantum fault-tolerant protocols are being implemented; otherwise, the error syndrome will keep accumulating and one will suffer from the so-called backlog problem~\cite{Terhal2015}.
Recently, a few computationally efficient (and provably correct) decoding algorithms have been developed for good QLDPC codes~\cite{leverrier2022sequential,gu2022efficient,dinur2022good}, assuming access to the noiseless error syndrome. 

To extract the error syndrome, one usually implements appropriate quantum circuits composed of basic quantum operations, such as state preparation, entangling gates and measurements. Unfortunately, these basic operations are imperfect and, for that reason, the assumption about the noiseless error syndrome is unrealistic. In particular, practical QEC codes and decoding algorithms should exhibit robustness to measurement errors. Arguably, one of the simplest ways to achieve such robustness involves repeating measurements until a reliable account of the error syndrome is obtained~\cite{Shor1996,dennisTopologicalQuantumMemory2002}. However, this approach incurs significant time overhead since the number of repetitions needed in general grows with the code distance.

An alternative to repeated measurement rounds of the error syndrome was introduced in the form of single-shot QEC by Bomb\'in~\cite{bombinSingleshotFaulttolerantQuantum2016}.
The basic idea behind single-shot QEC is to carefully select a code for which the decoding problem has sufficient structure to reliably infer qubit errors even with imperfect syndrome measurements. The strength of this approach is that significantly fewer measurements are necessary for codes that admit single-shot decoding compared to the simple strategy of repeated measurements. 

Single-shot QEC can be considered either for stochastic or adversarial noise. In the stochastic case, one is interested in noise that afflicts a (randomly selected) constant fraction of qubits. Additional structure may be needed for both the noise and the code, since the expected weight of the errors can be far beyond the code distance. Examples of such structure include sufficiently high expansion in the associated factor graphs, e.g., quantum expander codes~\cite{Fawzi_2018}; or the presence of geometrically local redundancies among constant-weight parity checks, e.g., the 3D subsystem toric code~\cite{SingleShotSubsystemToricCode,Bridgeman2023} and the gauge color code~\cite{bombin2015gauge}. In the adversarial case, as considered by Campbell~\cite{campbellTheorySingleshotError2019}, one can realize single-shot QEC for any code by measuring a carefully chosen set of parity checks; similar ideas of exploiting a redundant set of parity checks to simultaneously handle measurement and qubit errors were also explored in Refs.~\cite{Fujiwara2014,Ashikhmin2020,Delfosse2022}. The limitation of this approach is that, even when starting with a QLDPC code, the parity checks needed for single-shot QEC may have weight growing with code length, which makes it less appealing from the perspective of quantum fault-tolerant protocols.

We remark that while stochastic noise and adversarial noise models are generally incomparable, the distinction fades for asymptotically good QEC codes. Since these codes, by definition, have constant relative distance, they have the ability to correct arbitrary errors of weight up to a constant fraction of the number of qubits. In particular, stochastic noise with sufficiently low rate is correctable with high probability. Since in the rest of the paper we focus on good QLDPC codes, it suffices to consider the case of adversarial noise.

\subsection{Main Results}

In this article, we focus on a class of asymptotically good QLDPC codes called quantum Tanner codes~\cite{quantumTannerCodes}. They admit computationally efficient decoding algorithms, such as the sequential and parallel mismatch decomposition algorithms introduced in Ref.~\cite{leverrier2022parallel} and the potential-based decoder introduced in Ref.~\cite{gu2022efficient}. The problem of decoding quantum Tanner codes has so far been considered only in the scenario with noiseless error syndrome. Here, we study the performance of the aforementioned sequential and parallel mismatch decomposition decoders in the presence of measurement errors. We show that the decoders are \emph{single-shot}, under the following definition. For a more detailed discussion of single-shot decoding, see Section~\ref{sec:singleshot}.

Suppose a data error $e$ occurs on the qubits. Let $\sigma$ be the (ideal) syndrome corresponding to the data error. Suppose that the measured syndrome is corrupted by measurement error $D$. With access to the noisy syndrome $\td \sigma = \sigma + D$ as input, the decoder tries to output a correction $\hat f$ close to the data error.

\begin{definition}[Informal Statement of Definition~\ref{def:singleshotproblem}]
    A decoder is said to be $(\alpha,\beta)$-single-shot if, for sufficiently low-weight errors, the correction $\hat{f}$ returned on input $\tilde{\sigma}$ satisfies $|e + \hat f|_R \le \alpha|e|_R + \beta|D|$, where $|e|_R$ is the stabilizer-reduced weight of $e$, i.e., the weight of the smallest error equivalent to $e$ up to the addition of stabilizers.
\end{definition}

In other words, using a single round of noisy syndrome measurement, the decoder finds and applies the correction $\hat f$, resulting in the residual error $e+\hat f$ of weight below $\al|e|_R + \be|D|$.
Let $n$ be the number of physical qubits of the quantum Tanner code. Our main theorems are as follows.

\begin{theorem}[Informal Statement of Theorem~\ref{thm:main}]
    There exists a constant $\beta$ such that the sequential decoder (Algorithm~\ref{alg:seq_decoder}) is $(\alpha=0,\beta)$-single-shot.
\end{theorem}

\begin{theorem}[Informal Statement of Theorem~\ref{thm:parallel_residual}]~\label{thm:informal_parallel}
    There exists a constant $\beta$ such that for all $\alpha > 0$, the $O(\log(1/\alpha))$-iteration parallel decoder (Algorithm~\ref{alg:noisyparalleldecoder}) is $(\alpha,\beta)$-single-shot. In particular, for $O(\log n)$ iterations of parallel decoding one obtains $\alpha=0$.
\end{theorem}

We further consider the situation where multiple rounds of qubit error, noisy syndrome measurement, and decoding occur. We show that under mild assumptions on the weights of qubit and measurement errors, repeated applications of an $(\alpha,\beta)$-single-shot decoder will keep the residual error weight bounded. Specifically, consider the case where an initial error $(e_1,D_1)$ is partially corrected by the decoder, leaving a residual error $e_1'$. A new error $(e_2,D_2)$ is then applied on top of the existing residual error, giving total error $(e_1'+e_2,D_2)$. The decoder attempts to correct using a new round of syndrome measurements (without using the syndromes of previous rounds), leaving residue $e_2'$. This process is repeated for multiple rounds. Then we have the following.

\begin{theorem}[Informal Statement of Theorem~\ref{thm:multiroundadvserialerror}]
    Consider an $(\alpha,\beta)$-single-shot decoder and multiple rounds of errors $(e_i,D_i)$ for $i=1,\cdots,M$.
    For any $c > 0$, there exists a constant $C_*>0$ such that if $\max(|e_i|,|D_i|) \le C_*n$ for all $i$, then the final residual error $e_M'$ satisfies $|e_M'|_R\le cn$.
\end{theorem}

A direct implication of this result is that for the parallel decoder (Algorithm~\ref{alg:noisyparalleldecoder}), a constant number of iterations suffices to keep the residual error weight bounded at each round. This process can be repeated essentially indefinitely until ideal error correction is required, at which point the $O(\log n)$-iteration parallel decoder can be used. For more details, see the discussion at the end of Section~\ref{sec:MultipleRoundsOfDecoding}.

The rest of this paper is organized as follows. In Section~\ref{sec:qTanner}, we provide the necessary background on quantum Tanner codes. For more detailed explanations, see Refs.~\cite{quantumTannerCodes}~and~\cite{leverrier2022parallel}. In Section~\ref{sec:singleshot}, we describe the decoding problem for quantum (CSS) codes under measurement noise, and discuss the notion of single-shot decoding. We then define $(\alpha,\beta)$-single-shot decoding and derive general consequences of this definition under multiple rounds of error and decoding. The main result of this section is the proof of Theorem~\ref{thm:multiroundadvserialerror}. Section~\ref{sec:proofs} forms the bulk of the paper. There, we review the sequential and parallel decoders from Ref.~\cite{leverrier2022parallel} and prove that the decoders are single-shot in Theorems~\ref{thm:main} and~\ref{thm:parallel_residual}. Finally, we end with some discussions in Section~\ref{sec:discussion}.

\section{Quantum Tanner Codes}\label{sec:qTanner}

\subsection{Classical codes}
A classical binary linear code is a subspace $C\subseteq \bb F_2^n$. We refer to $n$ as the block length of the code. The number of encoded bits (also referred to as the code dimension) is given by $k=\dim C$ and the rate of the code is $R = k/n$. The distance of $C$ is defined as $d=\min_{x\in C\setminus\{0\}} |x|$, where $|\cdot|$ is the Hamming weight of a vector and where $0$ denotes the zero vector. A code with distance $d$ can protect against any unknown error of weight less than $d/2$. Often, it is useful to specify a code $C$ via a parity check matrix $H$. By definition, $C = \ker H$.

The dual code of a code $C$ is defined as $C^\perp = \{x\in \bb F_2^n: \langle x,y\rangle=0\ \forall y\in C\}$. The tensor product code of two codes $C_A\subseteq \bb F_2^A, C_B\subseteq \bb F_2^B$ is $C_A\otimes C_B\subseteq \bb F_2^{A\times B}$, where the codewords can be thought of as matrices such that every column is a codeword of $C_A$ and every row is a codeword of $C_B$. The dual tensor code of $C_A$ and $C_B$, denoted by $C_A \boxplus C_B$, is defined as
\[
C_A \boxplus C_B \equiv \left(C_A^\perp \otimes C_B^\perp\right)^\perp = C_A\otimes \bb F_2^B + \bb F_2^A\otimes C_B \subseteq \bb F_2^{A\times B}.
\]

A parity check matrix for $C_A\boxplus C_B$ is $H_A\otimes H_B$, where $H_A$ and $H_B$ are the parity check matrices of $C_A$ and $C_B$, respectively.

The dual tensor codes we use are required to satisfy the following robustness condition.
\begin{definition}
The code $C_A\boxplus C_B$ is said to be $\kappa$-product-expanding if any $x\in C_A\boxplus C_B$ can be expressed as $c+r$, with $c\in C_A\otimes \bb F_2^B$ and $r\in \bb F_2^A\otimes C_B$ such that
\begin{equation}
	\kappa \left(\frac{1}{|A|}\|c\|_A + \frac{1}{|B|}\|r\|_B\right)\le \frac{1}{|A||B|}|x|\, .
\end{equation}
\end{definition}

Here, $\|c\|_A$ denotes the number of non-zero columns in $c$ and $\|r\|_B$ denotes the number of non-zero rows in $r$. When it is clear from context, we will drop the subscripts on the norms. The notion of product-expansion was introduced by Panteleev and Kalachev~\cite{PK21}. It is equivalent to robust testability of tensor product codes~\cite{BS06} and agreement testability~\cite{dinur2021locally}, and also implies another notion called $w$-robustness of dual tensor codes~\cite{quantumTannerCodes}. It has been proven that random codes are product-expanding with high probability~\cite{KPTwoSided, dinur2022good}.

\begin{theorem}[Theorem 1 in Ref.~\cite{KPTwoSided}]
    Let $\rho\in (0,1)$. For any $\Delta$, let $C_A$ be a random code of dimension $\lceil\rho \Delta\rceil$ and $C_B$ be a random code of dimension $\lceil (1-\rho)\Delta\rceil$. There exists a constant $\kappa$ such that both $C_A\boxplus C_B$ and $C_A^\perp \boxplus C_B^\perp$ are $\kappa$-product-expanding with probability approaching 1 as $\Delta\to\infty$.
\end{theorem}

\subsection{Quantum codes}\label{subsec:qcodes}
An $n$-qubit quantum code is a subspace $\mc C$ of an $n$-qubit Hilbert space, i.e., $\mc C\subseteq \left(\bb C^2\right)^{\otimes n}$. We are interested in stabilizer codes, which are codes that can be expressed as the simultaneous $+1$-eigenspace of an abelian subgroup $\mc S$ of the $n$-qubit Pauli group satisfying $-I \not\in\mc S$. If $\mc S$ can be generated by two sets $\mc S_X$ and $\mc S_Z$ comprising, respectively, Pauli $X$-type and $Z$-type operators, then we refer to the corresponding stabilizer code as a Calderbank-Shor-Steane (CSS) code~\cite{cssCalderbankShor,cssSteane}. By ignoring the phase factors for such $X$-type and $Z$-type operators, we can identify them with their supports as vectors in $\bb F_2^n$.

For any CSS code stabilized by $\mc S=\langle \mc S_X, \mc S_Z\rangle$, we can define two $n$-bit classical codes $C_X = \ker H_X$ and $C_Z = \ker H_Z$, where each row in $H_X$ and $H_Z$ is the support of a stabilizer generator in $\mc S_X$ and $\mc S_Z$, respectively. The dimension of a CSS code is $k = k_X+k_Z-n$, where $k_X$ and $k_Z$ are the dimensions of $C_X$ and $C_Z$, respectively. The distance is $d = \min(d_X, d_Z)$, where $d_X = \min_{x\in C_Z\setminus C_X^\perp} |x|$ and $d_Z = \min_{x\in C_X\setminus C_Z^\perp} |x|$. A quantum code of distance $d$ can protect against any unknown error of weight less than $d/2$. A quantum code $\mc C\subseteq (\bb C^2)^{\otimes n}$ of dimension $k$ and distance $d$ is said to be an $[[n,k,d]]$ code. A family of CSS codes is said to be low-density parity-check (LDPC) if $H_X$ and $H_Z$ are sparse, i.e., have at most a constant number of non-zero entries in every column and row.

\subsection{Quantum Tanner code construction}
We now describe the construction of quantum Tanner codes. The code is placed on a geometric object called the left-right Cayley complex. Let $G$ be a finite group and $A=A^{-1}, B=B^{-1}$ be two symmetric generating sets of $G$. The left-right Cayley complex $\Cay_2(A,G,B)$ is a two-dimensional object with vertices $V$, edges $E$, and faces $Q$ defined as follows:
\begin{itemize}
	\item $V = V_{00}\sqcup V_{01}\sqcup V_{10}\sqcup V_{11}$, where $V_{ij} = G\times \{(i,j)\}$ for $i,j\in \{0,1\}$\,,
	\item $E = E_A\sqcup E_B$, where $E_A = \{\{(g,i0), (ag,i1)\}: g\in G, a\in A, i\in \{0,1\}\}$ and $E_B = \{\{(g,0j), (gb, 1j)\}: g\in G, b\in B, j\in \{0,1\}\}$\,,
	\item $Q = \{\{(g,00), (ag, 01), (gb, 10), (agb, 11)\}: g\in G, a\in A, b\in B\}$\,.
\end{itemize}
Let $Q(v)$ denote the set of faces incident to a given vertex $v$. Each face incident to $v$ can be obtained by choosing an $A$-type edge and a $B$-type edge incident to $v$ and completing them into a square. Therefore, $Q(v)$ is in bijection with the set $A\times B$, and can be thought of as a matrix with rows indexed by $A$ and columns indexed by $B$ (Figure~\ref{fig:local_view}). Similarly, the set of faces incident to a given $A$-edge is in bijection with $B$ and the set of faces incident to a given $B$-edge is in bijection with $A$.

\begin{figure}
    \centering
    \includegraphics[width=.8\linewidth]{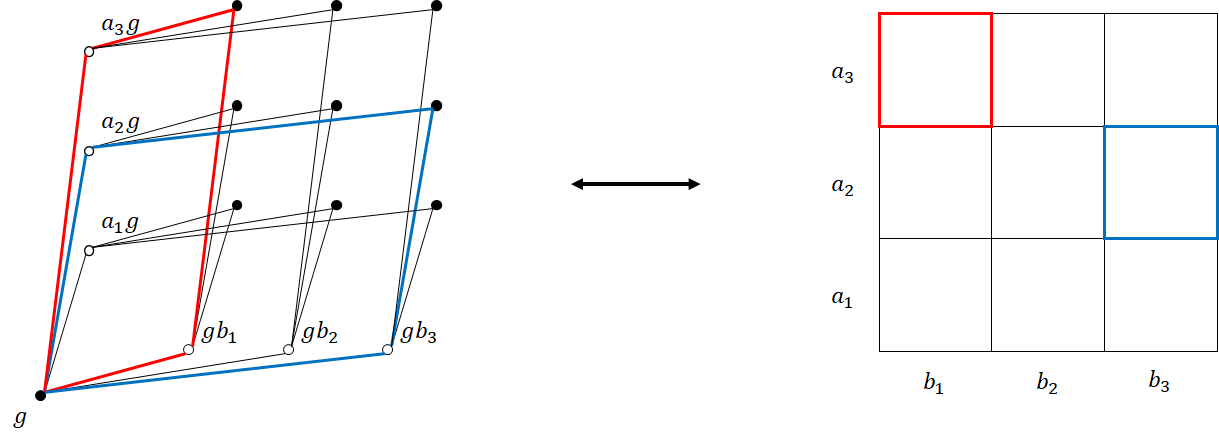}
    \caption{The local structure of the left-right Cayley complex around a vertex labelled by $g\in G$. The incident faces $Q(v)$ has a natural bijection with $A\times B$. As examples, the red and blue faces in the complex are mapped to the squares of the same colors in the matrix given by $A\times B$.}
    \label{fig:local_view}
\end{figure}

Consider the usual Cayley graph $\Cay(A,G)$ with the vertex set $G$ and the edge set $\{\{g,ag\}: g\in G, a\in A\}$. Ignoring the $B$ edges from the complex, we have that $(V,E_A)$ is the disjoint union of two copies of the bipartite cover of $\Cay(A,G)$. Similarly, $(V,E_B)$ is the disjoint union of two copies of the bipartite cover of $\Cay(G,B)$.\footnote{We denote the Cayley graph with left group action by $\Cay(A,G)$ and the Cayley graph with right group action by $\Cay(G,B)$. Note that the right Cayley graph $\Cay(G,B)$ with edges $\{g,gb\}$ is isomorphic to the left Cayley graph $\Cay(B,G)$ by mapping every $g$ to $g^{-1}$.} We say that a $\Delta$-regular graph is Ramanujan if the second largest eigenvalue of its adjacency matrix is at most $2\sqrt{\Delta-1}$, and we will consider left-right Cayley complexes with component Cayley graphs $\Cay(A,G)$ and $\Cay(G,B)$ that are Ramanujan. Explicitly, Ramanujan Cayley graphs can be obtained by taking $G = \operatorname{PSL}_2(q^i)$, where $q$ is an odd prime power and $A, B$ are (appropriately chosen) symmetric generating sets of constant size $\Delta = |A|=|B|=q+1$~\cite{dinur2021locally}.

Quantum Tanner codes are CSS codes defined by placing qubits on the faces of a left-right Cayley complex. We fix two classical codes, $C_A$ of length $|A|$ and $C_B$ of length $|B|$, which are used to define a pair of local codes providing the parity checks of the quantum code. An $X$-type stabilizer generator is defined as a codeword from a generating set of $C_0 = C_A\otimes C_B$, with support on the faces incident to a given vertex in $V_0 = V_{00}\cup V_{11}$. More precisely, there is an $X$-type stabilizer generator $s(x,v)$ for every generator $x\in C_A\otimes C_B$ and every vertex $v\in V_0$. Identifying $Q(v)$ with $A\times B$ using the bijection explained earlier, the support of $s(x,v)$ is the subset of $Q(v)$ defined by the support of $x$; see Figure~\ref{fig:stabilizer_example} for an illustration. Similarly, the $Z$-type stabilizers are generated by codewords of $C_1 = C_A^\perp\otimes C_B^\perp$ on the faces incident to vertices of $V_1 = V_{01}\cup V_{10}$. The fact that $X$ and $Z$ parity checks commute is because $X$ and $Z$ generators are either disjoint or overlap on the faces incident to a single edge. On this set of faces, isomorphic to either $B$ or $A$, the supports of the $X$ and $Z$ operators are codewords of either $C_B$ and $C_B^\perp$, respectively, or $C_A$ and $C_A^\perp$, respectively. It is clear that a family of quantum Tanner codes is QLDPC if the degrees of the component Cayley graphs are bounded.

\begin{figure}
    \centering
    \includegraphics[width=.8\linewidth]{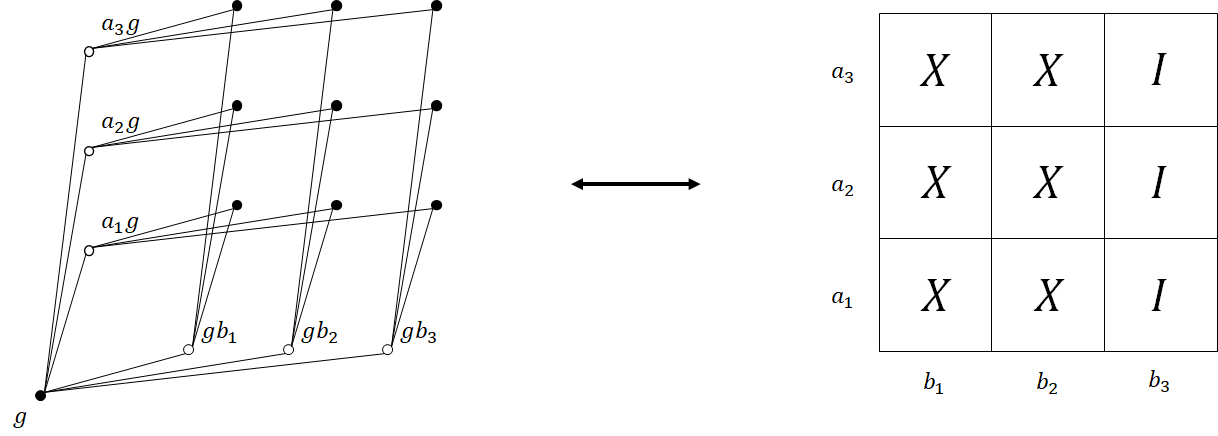}
    \caption{An example of a stabilizer generator with local codes $C_A = \spn\{111\}$ and $C_B=\spn\{110,011\}$. The codeword $x=111\otimes 110 \in C_A\otimes C_B$ has support as shown on the right. Identifying that matrix with the faces incident to a $V_0$ vertex gives an $X$-type stabilizer generator.}
    \label{fig:stabilizer_example}
\end{figure}

Leverrier and Z\'emor showed that quantum Tanner codes defined on expanding left-right Cayley complexes using product-expanding local codes have good parameters~\cite{quantumTannerCodes,leverrier2022parallel}.

\begin{theorem}[Theorem 1 in Ref.~\cite{leverrier2022parallel}]
    \label{thm:codeparameters}
    Let $\rho, d_r, \kappa\in (0,1)$ and $\Delta$ be a sufficiently large constant. Let $C_A, C_B\subseteq \bb F_2^\Delta$ be classical codes of rates $\rho$ and $(1-\rho)$ respectively, such that the distances of $C_A, C_B, C_A^\perp, C_B^\perp$ are all at least $d_r\Delta$, and such that $C_A\boxplus C_B$ and $C_A^\perp\boxplus C_B^\perp$ are both $\kappa$-product-expanding. Using a family of $\Delta$-regular Ramanujan Cayley graphs $\Cay(A,G)$ and $\Cay(G,B)$, define the left-right Cayley complex $\Cay_2(A,G,B)$. Then the quantum Tanner codes defined using the components above have parameters
    \begin{equation}
        \left[\!\left[n,k\ge (1-2\rho)^2n, d\ge \frac{d_r^2\kappa^2}{256\Delta}n\right ]\!\right]\, .
    \end{equation}
\end{theorem}

\section{Single-shot Decoding}\label{sec:singleshot}
\subsection{Decoding CSS codes}\label{subsec:decodeCSS}

Let us now formally define the decoding problem for quantum (CSS) codes. After we encode logical information in a quantum code, errors will occur on the physical system. We are interested in how to ``undo'' these errors and, subsequently, recover the original logical state. Specifically, consider a logical state $|\psi\rangle$ of a stabilizer code $\mc C$. A Pauli error $E$ occurs, and we gain information about the error by measuring a set of stabilizer generators $\{S_i\}$. This gives a syndrome $\sigma$, a bit string whose values $\sigma_i$ correspond to the eigenvalues $(-1)^{\sigma_i}$ of the stabilizers measured. Thus, $\sigma_i=0$ whenever $S_i$ commutes with $E$ and $\sigma_i=1$ when it anticommutes. The task of decoding is to use $\sigma$ to determine a correction $\hat{F}$ such that $\hat{F}E|\psi\rangle = |\psi\rangle$. In other words, $\hat{F}E$ should be a stabilizer of the code. When $\mc C$ is a CSS code, we can express the problem as follows.

\begin{definition}
	Let $\mc C$ be a CSS code specified by two parity check matrices $H_X\in \bb F_2^{r_X\times n}$ and $H_Z\in \bb F_2^{r_Z\times n}$. Let $e=(e_X, e_Z)\in \bb F_2^{2n}$ be an error with corresponding syndrome $\sigma = (\sigma_X, \sigma_Z)\in \bb F_2^{r_Z+r_X}$, where $\sigma_Z = H_Xe_Z$ and $\sigma_X = H_Ze_X$. Given input $\sigma$, the task of decoding is to find corrections $\hat{f} = (\hat{f}_X, \hat{f}_Z)\in \bb F_2^{2n}$ such that $e_X+\hat{f}_X\in C_X^\perp$ and $e_Z+\hat{f}_Z\in C_Z^\perp$.
\end{definition}

In the definition above, we associate the bit string $e = (e_X,e_Z)$ with the Pauli errors $E=E_XE_Z$ where $E_X$ and $E_Z$ are Pauli $X$ and $Z$ operators with support $e_X$ and $e_Z$, respectively (ignoring phase information). The correction $\hat{f}$ is similarly associated with a Pauli operator $\hat{F}$.

We note that for CSS codes, the decoding problem can be split into two separate problems for the $X$ and $Z$ codes that can be solved independently. For quantum Tanner codes in particular, there is symmetry between the $X$ and $Z$ codes, as can be seen by switching $V_0$ and $V_1$ labels and switching $C_A, C_B$ with $C_A^\perp, C_B^\perp$. Therefore, it suffices to give an algorithm for decoding one type of error. In the remainder of the paper, we will consider solely the case where $X$-errors occur, with $Z$-errors treated analogously. For convenience, we will often drop subscripts, for example writing $e$ for $e_X$ or $H$ for $H_Z$.

The above discussion assumes that the ideal syndrome is accessible to the decoder. Let us now consider the case when the syndrome measurements are unreliable, motivated by the fact that the quantum circuits implementing the parity checks are necessarily imperfect. Suppose that the ideal syndrome $\sigma_X$ of an error $e_X$ is corrupted by measurement error $D_X$, so that the actual noisy syndrome readout is $\td\sigma_X = \sigma_X + D_X$. A naive decoding of the syndrome $\td \sigma_X$ may result in a correction $\hat{f}_X$ which does not bring the state back to the code space, i.e., $e_X+\hat{f}_X \not\in C_X^\perp$. Furthermore, there may be no guarantee that $e_X+\hat{f}_X$ is close to $C_X^\perp$.

One of the standard procedures to account for measurement errors is to repeatedly measure the stabilizer generators in order to gain enough confidence in their measurement outcomes~\cite{Shor1996, dennisTopologicalQuantumMemory2002}.
This will incur large time overhead.
Alternatively, syndrome measurements can be performed fault-tolerantly by preparing special ancilla qubit states offline~\cite{Knill2005, Steane1997}.
This will incur large qubit overhead.
It would be ideal if we could avoid both overheads at the same time.

\subsection{Single-shot decoding}\label{subsec:singleshot}

Bomb\'in~\cite{bombinSingleshotFaulttolerantQuantum2016} introduced \emph{single-shot} decoders as an alternative approach. These decoders take in a noisy syndrome as input and, even in the presence of syndrome noise, return a correction that can be used to reduce the data error. Most likely, there will be some resulting residual error, but its weight is bounded by some function of the syndrome noise. In more detail, the single-shot property posits that it suffices to perform $O(n)$ parity check measurements (in the context of QLDPC codes, one further requires constant weight of measured parity checks), and, using \emph{only} these measurement outcomes, one can perform reliable QEC that keeps the residual noise at bay.

In our analysis, we need the following definition.

\begin{definition}
Let $\mc C$ be an $n$-qubit CSS code and $e\in\bb F_2^n$ be a Pauli $X$ error.
The stabilizer-reduced weight $|e|_R$ of $e$ is defined as the weight of the smallest error equivalent to $e$ up to the addition of stabilizers of $\mc C$, i.e., $|e|_R = \min_{e'\in C_X^{\perp}} |e+e'|$.
The stabilizer-reduced weight of a Pauli $Z$ error is defined analogously.
\end{definition}

The stabilizer-reduced weight of an error is a convenient theoretical measure of
how detrimental the error really is.
Note that since stabilizers do not change the code state, errors are only well-defined up to the addition of stabilizers. As such, any bound on the performance of the decoder is unambiguously defined using the stabilizer-reduced weight, which can be significantly smaller than the original weight.

Since we focus on asymptotically good QLDPC codes, it is enough to consider single-shot decoding for adversarial noise. Campbell~\cite{campbellTheorySingleshotError2019} captures adversarial single-shot decoding as follows. Let both the data error $e$ and the syndrome noise $D$ be sufficiently small. A decoder is single-shot if it outputs a correction such that the weight of the residual error is bounded by a polynomial of $|D|$. In this work, we would like to consider constant-time decoding using the parallel decoder (Algorithm~\ref{alg:noisyparalleldecoder}) for quantum Tanner codes. This setting does not directly fit into the previous definition since the residual error could depend on $|e|$ in addition to $|D|$. To allow for nontrivial dependence on $|e|$, we give the following definition, which is relevant for asymptotically good codes where the residual error size is at most linear in $|e|$ and $|D|$.

\begin{definition}\label{def:singleshotproblem}
    Let $\mc C$ be a CSS code specified by parity check matrices $H_X\in \bb F_2^{r_X\times n}$ and $H_Z\in \bb F_2^{r_Z\times n}$. Let $e=(e_X, e_Z)\in \bb F_2^{2n}$ be a data error, $D = (D_X, D_Z)\in \bb F_2^{r_Z+r_X}$ be a syndrome error, and $\td\sigma = (\td\sigma_X, \td\sigma_Z)\in \bb F_2^{r_Z+r_X}$ be the corresponding noisy syndrome, where $\td\sigma_X = H_Ze_X + D_X$ and $\td\sigma_Z = H_Xe_Z + D_Z$. A decoder for $\mathcal{C}$ is $(\al,\be)$-single-shot if there exist constants $A,B,C$ such that, for $P\in\{X,Z\}$,
    whenever
    \begin{align}
        A|e_P|_R + B|D_P| \leq Cn,
    \end{align}
    the decoder finds a correction $\hat{f}_P\in \bb F_2^{n}$ from given input $\td \sigma_P$ such that 
    \begin{align}
        |e_P+\hat{f}_P|_R\le \alpha|e_P| + \beta|D_P|\, .
    \end{align}
\end{definition}

This definition, combined with Theorems~\ref{thm:main} and~\ref{thm:parallel_residual} below, gives the following results for the sequential and parallel decoders of the quantum Tanner codes.

\begin{theorem}[Summary]
There exist constants $A,B,C,\beta > 0$ (dependent on the parameters of the quantum Tanner code) such that if $A|e|_R + B|D| \le Cn$, then the following conditions hold:
\begin{enumerate}
\item The sequential decoder (Algorithm~\ref{alg:seq_decoder}) is $(\alpha=0,\beta)$-single-shot.

\item The parallel decoder (Algorithm~\ref{alg:noisyparalleldecoder}) with $k$-iterations is $(\alpha = 2^{-\Omega(k)},\beta)$-single-shot.
\end{enumerate}
\end{theorem}

Note that the runtime of the sequential decoder is $O(n)$, and each iteration of the parallel decoder is constant time. For the parallel decoder, $\alpha$ decreases exponentially with the number of parallel decoding iterations $k$, and the results of this section will hold when $k$ is a sufficiently large constant. It suffices to take $k = O(\log n)$ for $\alpha = 0$ in the parallel decoder.

Finally, we remark that we may increase the robustness to measurement errors and improve the overall performance of single-shot decoding by leveraging redundancies among parity checks, similar to the ideas explored in Refs.~\cite{Fujiwara2014,Ashikhmin2020,Delfosse2022}.
We can apply this approach to quantum Tanner codes without compromising their QLDPC structure, which is a crucial difference between our setting and the aforementioned works. Specifically, stabilizer generators of quantum Tanner codes are supported on local neighborhoods, defined by the local codes $C_0$ and $C_1$.
We may apply the technique of adding redundancy to each set of local checks separately. Since the local codes are of length $\Delta^2$, any redundant check in a fixed local neighborhood will not have weight more than $\Delta^2$, which is comparable to the weight of the original checks.

\subsection{Multiple rounds of decoding}\label{sec:MultipleRoundsOfDecoding}

In this section, we discuss what happens after multiple rounds of errors, noisy measurements, and decoding. We show that under the assumptions of Definition~\ref{def:singleshotproblem}, there exists a variety of noise models such that, as long as the overall noise level is sufficiently small, the encoded quantum information will persist for an exponential number of rounds.

The results proven in this section hold for any decoder that can solve the single-shot decoding problem under Definition~\ref{def:singleshotproblem}. More precisely, we assume that if the decoder is given the noisy syndrome from data error $e\in \bb F_2^n$ and syndrome error $D\in \bb F_2^{r_Z}$ satisfying
\begin{equation}
    A|e|_R + B|D| \le Cn\,,
\end{equation}
then it outputs a correction $\hat f$ such that the residual error satisfies
\begin{equation}
    |e+\hat f|_R\le \alpha |e| + \beta |D|\, .
\end{equation}
We will assume that $\beta$ is constant and that $\alpha$ is a parameter in the decoder that can be made arbitrarily small. For our analysis, we let $R,S$ be constants such that
\begin{equation}
    R\le \frac{(1-\alpha)C}{2A} \quad \text{and} \quad S\le \frac{(1-\alpha)C}{2\left({\beta A}+(1-\alpha)B\right)}\, .\label{eq:ABCRSinequality}
\end{equation}

We prove that as long as the data and syndrome errors in each round are sufficiently small, the total error can be kept small indefinitely.

\begin{theorem}
    \label{thm:multiroundadvserialerror}
    Consider errors $(e_i,D_i)$ that occur on rounds $i=1, 2, \cdots$, with decoding in between each round using new syndrome measurements (i.e., without using the previous syndromes). If the errors satisfy $|e_i|\le Rn$ and $|D_i|\le Sn$ for every round $i$, then the residual error $e_i'$ after each round $i$ satisfies
    \begin{equation}
        \label{eq:totalerrorbound}
        |e_i'|_R\le \frac{\alpha R+\beta S}{1-\alpha}n\, .
    \end{equation}
\end{theorem}

\begin{proof}
    Initially, $e_0' = 0$, which satisfies the bound. Suppose after round $i-1$, the residual error $e_{i-1}'$ satisfies~\eqref{eq:totalerrorbound}. The new total error is $e_{i-1}' + e_i$, and we have
    \begin{align}
        A|e_{i-1}'+e_i|_R + B|D_i| &\le A|e_{i-1}'|_R + A|e_i| + B|D_i|\\
        &\le A\frac{\alpha R+\beta S}{1-\alpha}n + ARn + BS n\\
        &\le Cn\, ,
    \end{align}
    where the last inequality follows since
    \begin{equation}
    A\frac{R+\beta S}{1-\alpha} + BS \le C
    \end{equation}
    for $R$ and $S$ satisfying~\eqref{eq:ABCRSinequality}.Therefore, the decoder returns a correction $\hat f$ with residual error
    \begin{align}
        |e_i'|_R &\le \alpha |e_{i-1}' + e_i|_R + \beta |D_i|\\
        &\le \alpha |e_{i-1}'|_R + \alpha |e_i| + \beta |D_i|\\
        &\le \alpha \frac{\alpha R+\beta S}{1-\alpha}n + \alpha Rn + \beta Sn\\
        &= \frac{\alpha R+\beta S}{1-\alpha}n\, ,
    \end{align}
    where the third inequality uses the inductive hypothesis. 
\end{proof}

From this result, we can immediately analyze the stochastic setting in which large errors are unlikely.

\begin{corollary}
    \label{cor:multiroundstochasticerror}
    Let $\{(e_i, D_i)\}_{i=1}^M$ be randomly distributed data and syndrome errors (with possible correlations) such that
    \begin{equation} \label{eq:approximateadversarial}
        \Pr(|e_i| > Rn) \le e^{-an}, \quad \text{and} \quad \Pr(|D_i| > Sn) \le e^{-bn}\, ,
    \end{equation}
    for constants $a,b > 0$. Suppose the decoder is run after each round of errors using new syndrome measurements (i.e., without using the syndromes of previous rounds).  Then the final residual error $e_M'$ satisfies
    \begin{equation}
        \Pr\left(|e_M'|_R > \frac{\alpha R+\beta S}{1-\alpha}n\right) \le M(e^{-an} + e^{-bn})\, .
    \end{equation}
\end{corollary}
\begin{proof}
    This follows immediately from Theorem~\ref{thm:multiroundadvserialerror} after using a union bound on the probability of a large data or syndrome error at every round.
\end{proof}

As a sample application of Corollary~\ref{cor:multiroundstochasticerror}, we analyze the case of $p$-bounded noise~\cite{gottesman2014faulttolerant, Fawzi_2018}, although any model of errors with sufficiently suppressed tails will give the same conclusions.
 
\begin{definition}[$p$-bounded noise]\label{def:pbounded}
Let $p \in [0,1)$. Let $A$ be a set and let $2^A$ be its power set. We say that a probability distribution $E: 2^A \rightarrow [0,1]$ is \emph{$p$-bounded} if for any $B \subseteq A$ we have
\begin{align}
\sum_{B' \supseteq B}E(B') \le p^{|B|}\ .
\end{align}
\end{definition}

\begin{corollary} Let $\{(e_i,D_i)\}_{i=1}^M$ be data and syndrome errors where each of the marginal distributions of $e_i$ and $D_i$ are $p$- and $q$-bounded, respectively. Suppose the decoder is run after each round of errors using a new round of syndrome measurements (without using the syndromes of previous rounds). Then, the final residual error $e_M'$ satisfies
\begin{align}
    \Pr\left(|e_M'|_R > \frac{\alpha R+\beta S}{1-\alpha}n\right) \le M\left(e^{-n\ln (2^{-H(R)}p^{-R})} +e^{-n\ln (2^{-\varrho H(S/\varrho)}q^{-S})}\right) \ ,
\end{align}
where $H(\tau) = -\tau\log_2 \tau - (1-\tau)\log_2(1-\tau)$ is the binary entropy function, and $\varrho = r_Z/n$.
\end{corollary}
\begin{proof}
Let us first upper bound $\Pr(|e_i| > R n)$. We have
\begin{align}
        \Pr(|e_i| > R n) &= \sum_{\substack{|e|> Rn}} \Pr(e_i=e)\le \sum_{|e|= Rn}\Pr(e_i\supset e)
        \le \sum_{|e|= Rn}p^{|e|} \le \binom{n}{R n}p^{Rn}\, ,
    \end{align}
    where the last inequality follows by $p$-boundedness. Using the binary entropy bound for the binomial coefficient, we then have
    \begin{align}
        \Pr(|e_i| > R n) &\le \binom{n}{R n}p^{Rn}
        \le 2^{nH(R)}p^{Rn}
        =e^{-n\ln (2^{-H(R)}p^{-R})}\, .    
    \end{align}
    Similarly, we have
    \begin{align}
        \Pr(|D_i|> S n) \le e^{-n\ln (2^{-\varrho H(S/\varrho)}q^{-S})}\, .
    \end{align}
    Applying  Corollary~\ref{cor:multiroundstochasticerror} gives the result.
\end{proof}

In particular, there exist thresholds $(p_*, q_*) = (2^{-H(R)/R}, 2^{-\varrho H(S/\varrho)/S})$ below which errors are kept under control for an exponential number of rounds of single-shot QEC with high probability.

Finally, we comment on the last round of QEC. In a typical setting of fault tolerance, we choose to measure logical qubits in the computational basis, which for a CSS code can be accomplished by measuring each physical qubit (also in the computational basis).
We then apply one final round of QEC, where the $Z$-stabilizer eigenvalues are inferred by multiplying the $Z$-measurement outcomes from those qubits in the stabilizer supports. Note that in this final round, any measurement error can be treated as an $X$ data error immediately before the measurement. We run the decoder with $\alpha$ sufficiently small so that by the guarantee on the decoder, $|e+\hat f|_R=0$, i.e., we completely correct the error. We can then infer the logical information by combining the corrected single-qubit $Z$-measurement outcomes making up the $Z$-logical operators. Therefore, fault tolerance may be achieved by using a faster (e.g., constant-time) decoder with larger $\alpha$ value in the middle of the computation, and only applying the full decoder (e.g., logarithmic-time) with $\alpha=0$ at the end of the computation.

\section{Proofs of Single-shot Decoding of Quantum Tanner Codes}\label{sec:proofs}

\subsection{Decoding Algorithms}
We consider the decoding problem for quantum Tanner codes with parameters as in Theorem~\ref{thm:codeparameters}.  We first provide an overview of how the decoder works. As before, we will work exclusively with $X$-type errors, with $Z$-errors being analogous. Suppose that the code state experiences data error $e$, and the measurements experience syndrome error $D$. The decoder is consequently given as input the noisy syndrome $\td\sigma = \sigma + D = H_Ze + D$. Due to the structure of the code, the global syndrome $\td\sigma$ can equivalently be viewed as a set of noisy local syndromes $\{\td\sigma_v\}_{v\in V_1}$, where $\td \sigma_v$ denotes the restriction of $\td \sigma$ to the checks associated with the local code $C_1^\perp$ at vertex $v$. At each $V_1$ vertex, the decoder computes a minimal weight correction $\td\ep_v\subseteq Q(v)$ based on the local syndrome $\td\sigma_v$, i.e,
\begin{align}
\td\ep_v = \operatorname{arg min} \{|y|: y\subseteq Q(v), \sigma_v(y)=\td\sigma_v\}\, .
\end{align}
Note that this is a completely local operation which can be done without consideration of the syndrome state of the other vertices. Each square $q\in Q$ contains two $V_1$ vertices, say $v\in V_{01}$ and $v' \in V_{10}$. These two vertices are each associated with their own local corrections, $\td\ep_v$ and $\td\ep_{v'}$, which may disagree on whether there is an error on $q$. If there is no disagreement on any square $q\in Q$, then a global correction $\hat f\in \bb F_2^Q$ can be unambiguously defined by 
\begin{align}
\hat f = \bigsqcup_{v\in V_{01}}\td\ep_v = \bigsqcup_{v'\in V_{10}}\td\ep_{v'}.
\end{align}
However, this will usually not be the case. The disagreement between the different candidate local corrections is captured by a ``noisy mismatch vector'' defined as
\begin{equation}
	\label{eq:noisymismatchdef}
	\td Z = \sum_{v\in V_1}\td\ep_v\, .
\end{equation}

The goal of the main part of the algorithm is to reduce the size of $\td Z$ by successively updating the best local corrections on the $V_1$ vertices. For example, it is possible that for a given $v\in V_1$, replacing $\td\ep_v$ with $\td\ep_v + x$ for some $x\in C_1^\perp$ in~\eqref{eq:noisymismatchdef} would significantly decrease $|\tilde Z|$. In general, we attempt to decompose $\td Z$ by adding codewords $x\in C_1^\perp$ on local views $Q(v)$ of vertices $v\in V$.\footnote{In the presence of measurement errors, a full decomposition of $\tilde{Z}$ into local codewords may not be possible. See Definition~\ref{def:mismatch} and related comments before and after.} We keep track of the decomposition process through quantities $\hat C_0, \hat C_1, \hat R_0, \hat R_1\subseteq \bb F_2^Q$, which are initially 0 and updated as follows. Suppose $x = c + r$ is supported on a $V_{ij}$ local view ($i,j\in \{0,1\}$), where $c\in C_A\otimes \bb F_2^B$ and $r\in \bb F_2^A\otimes C_B$. Then we add $c$ to $\hat C_j$ and $r$ to $\hat R_i$. The interpretation is that $\hat C_1 + \hat R_0$ is the total change made to the local corrections $\td\ep_v$ from the $V_{01}$ vertices, and $\hat C_0 + \hat R_1$ is the total change made to those from the $V_{10}$ vertices. Therefore, at the end of the procedure, we output a guess for the error, which from the perspective of the $V_{01}$ vertices is
\begin{equation}
	\hat f = \sum_{v\in V_{01}}\td \ep_v + \hat C_1 + \hat R_0\, .
\end{equation}

The algorithm can run either sequentially (Algorithm~\ref{alg:seq_decoder}) or in parallel (Algorithm~\ref{alg:noisyparalleldecoder}), with the corresponding $\td Z$ decomposition subroutines presented in Algorithm~\ref{alg:seq_mismatch} and Algorithm~\ref{alg:paralleldecoder} respectively.

\begin{algorithm}[H]
	\caption{Sequential decoder for quantum Tanner codes with parameter $\varepsilon$}
	\label{alg:seq_decoder}
	\begin{flushleft}
		\textbf{Input:} A noisy syndrome $\tilde\sigma$ arising from data error $e$ and syndrome error $D$.\\
		\textbf{Output:} A correction $\hat{f}$ that approximates $e$.
		\begin{algorithmic}[1]
			\State $\td\ep_v\gets \operatorname{arg min} \{|y|: y\subseteq Q(v), \sigma_v(y)=\td\sigma_v\}$ (or $\td\ep_v\gets 0$ if no such $y$ exists) for all $v\in V_1$
			\State $\td Z\gets \sum_{v\in V_1}\td\ep_v$
			\State $(\hat{C}_0,\hat{C}_1,\hat{R}_0,\hat{R}_1) \gets \textsc{Mismatch}_\varepsilon(\tilde{Z})$
			\State $\hat{f} \gets\sum_{v\in V_{01}}\td\ep_v + \hat{C}_1 + \hat{R}_0$
			\State \Return $\hat{f}$
		\end{algorithmic}
	\end{flushleft}
\end{algorithm}

\begin{algorithm}[H]
	\caption{Sequential mismatch decomposition with parameter $\varepsilon$}
	\label{alg:seq_mismatch}
	\textbf{Input:} A vector $Z \in \mathbb{F}_2^Q$.\\
	\textbf{Output:} A collection $(\hat{C}_0,\hat{C}_1, \hat{R}_0, \hat{R}_1) \equiv \textsc{Mismatch}_\varepsilon(Z)$.
	\begin{algorithmic}[1]
    \State Set $\hat{C}_0 = \hat{C}_1 = \hat{R}_0 = \hat{R}_1 = 0$ and $\hat{Z} = Z$.
    \While {$\hat{Z} \neq 0$} 
      \If{$\exists v \in V_{ij}$ and $0\neq x_v\in C_1^\perp$ in $Q(v)$ such that $|\hat{Z}| - |\hat{Z} + x_v| \ge (1-\ep)|x_v|$}
      \LongState{Find $r_v\in \bb F_2^A\otimes C_B$ and $c_v\in C_A\otimes \bb F_2^B$ such that $\|c_v\| + \|r_v\|$ is minimal among \newline all $c_v, r_v$ such that $r_v + c_v = x_v$}
      \State $\hat{C}_j \gets \hat{C}_j + c_v$
      \State $\hat{R}_i \gets \hat{R}_i + r_v$
      \State $\hat{Z} \gets \hat{Z} + c_v + r_v$
      \Else
      \State \Return ($\hat{C}_0,\hat{C}_1,\hat{R}_0,\hat{R}_1$)
      \EndIf
      \EndWhile
      \State \Return $(\hat{C}_0,\hat{C}_1, \hat{R}_0, \hat{R}_1)$
    \end{algorithmic}
\end{algorithm}

\begin{algorithm}[H]
	\caption{Parallel decoder for quantum Tanner codes with $k$ iterations}
	\label{alg:noisyparalleldecoder}
		\textbf{Input:} A noisy syndrome $\tilde\sigma$ from a data error $e$ and syndrome error $D$, and an integer $k>0$.\\
		\textbf{Output:} A correction $\hat{f}$ that approximates $e$.
		\begin{algorithmic}[1]
            \ParallelForEach{$v \in V_{1}$}
                \State $\td\ep_v\gets \operatorname{arg min} \{|y|: y\in Q(v), \sigma_v(y)=\td\sigma_v\}$ (or $\td\ep_v\gets 0$ if no such $y$ exists)
                \State $\td Z\gets \sum_{v\in V_1}\td\ep_v$
                \State $\hat{f} \gets \sum_{v\in V_{01}}\td\ep_v$
            \EndParallelForEach
            \State $(\hat{C}_0,\hat{C}_1,\hat{R}_0,\hat{R}_1)\gets \textsc{ParMismatch}^{(k)}(\tilde{Z})$
            \State $\hat{f} \gets \hat{f} + \hat{C}_1 + \hat{R}_0$ \quad\quad {\color{gray} // update $\hat{f}$ in parallel for each vertex $v\in V_{01}$}
            \State \Return $\hat{f}$
		\end{algorithmic}
\end{algorithm}

\begin{algorithm}[H]
	\caption{Parallel mismatch decomposition procedure with $k$ iterations}
	\label{alg:paralleldecoder}
		\textbf{Input:} A vector $Z \in \mathbb{F}_2^Q$ and integer $k > 0$. \\
		\textbf{Output:} A collection $(\hat{C}_0,\hat{C}_1, \hat{R}_0, \hat{R}_1) \equiv \textsc{ParMismatch}^{(k)}(Z)$.
        
		\begin{algorithmic}[1]
            \State Set $\hat{C}_0 = \hat{C}_1 = \hat{R}_0 = \hat{R}_1 = 0$ and $\hat{Z} = Z$.
            \RepeatN{$k$}
            \For{$(i,j) \in \{0,1\}^2$}
            \ParallelForEach{$v \in V_{ij}$}
            \If{there exists $0\neq x_v\in C_1^\perp$ in $Q(v)$ such that $|\hat{Z}| - |\hat{Z} + x_v| \geq |x_v|/2$ }
            \State Choose $x_v$ such that $|x_v|$ maximal among all possible choices 
            \LongState{Find $r_v\in \bb F_2^A\otimes C_B$ and $c_v\in C_A\otimes \bb F_2^B$ such that $\|c_v\| + \|r_v\|$ is minimal \\ among all $c_v, r_v$ such that $r_v + c_v = x_v$}
            
            \State $\hat{C}_j \gets \hat{C}_j + c_v$
            \State $\hat{R}_i \gets \hat{R}_i + r_v$
            \State $\hat{Z} \gets \hat{Z} + c_v + r_v$

            \EndIf
            \EndParallelForEach
            \EndFor
            \EndRepeatN
            \State \Return $(\hat{C}_0,\hat{C}_1, \hat{R}_0, \hat{R}_1)$
		\end{algorithmic}
\end{algorithm}

These algorithms were analyzed in the scenario with perfect measurement outcomes in Ref.~\cite{leverrier2022parallel}, giving the following results:

\begin{theorem}[Theorem 13 in Ref.~\cite{leverrier2022parallel}]\label{thm:lev1}
    Let $\ep\in (0,1)$. Suppose Algorithm~\ref{alg:seq_decoder} with parameter $\ep$ is given as input the noiseless syndrome $\sigma = H_Ze$ of an error $e\in \bb F_2^Q$ of weight
    \begin{equation}
        |e| \le \frac{1}{2^{11}}\min\left(\frac{\ep^3}{16},\kappa\right)(1-\ep)d_r^2\kappa^2\frac{n}{\Delta}\, .
    \end{equation}
    Then it will output a correction $\hat f$ such that $e+\hat f\in C_X^\perp$  in time $O(n)$.
\end{theorem}

\begin{theorem}[Theorem 20 in Ref.~\cite{leverrier2022parallel}]
    Let $\ep\in (0,1/6)$. Suppose Algorithm~\ref{alg:noisyparalleldecoder} is given as input the noiseless syndrome $\sigma = H_Ze$ of an error $e\in \bb F_2^Q$ of weight
    \begin{equation}
        |e| \le \frac{1}{2^{12}}\min\left(\frac{\ep^3}{16},\kappa\right)d_r^2\kappa^2\frac{n}{\Delta}\, .
    \end{equation}
    Then it will output a correction $\hat f$ such that $e+\hat f\in C_X^\perp$ in time $O(\log n)$.
\end{theorem}

In the next sections, we will consider what happens when the decoders are given a syndrome with possible errors.

\subsection{Proof Preliminaries}

We first give a summary of the main ideas of the proof. The key idea of the proof is to bound the reduction in the weight of the noisy mismatch vector $\tilde{Z}$ through each step of the algorithm, and to show that when the weight of $\tilde{Z}$ is reduced, the weight of the residual error is also subsequently reduced. There is a technical challenge to this idea however: there is no direct relation between the weight of $\tilde{Z}$ and the error weight.

To bridge these two objects, we define the notion of an ideal mismatch vector $Z$ (see Eq.~\eqref{eq:noiseless_def} below), which is equal to $\tilde{Z}$ when there is no measurement noise. Since the mismatch $Z$ only captures the portion of the error which cannot be removed using independent local corrections, we must first ``pre-process'' the error by making any possible local corrections (see Eq.~\eqref{eq:pre-process} below). This establishes a direct connection between $Z$ and the ``pre-processed'' error $e_0$ (see Lemma~\ref{lemma:V10weightedmismatchZ0}) and our analysis will be built upon this connection.

We show that if $Z$ is decomposable into local corrections by Algorithm~\ref{alg:seq_mismatch}, then most of these correction sets will also reduce the weight of $\tilde{Z}$ (Lemma~\ref{lem:flipexists}). This in turn allows us to relate the weights of $Z$ and $\tilde{Z}$. Finally, we show that if the qubit and measurement error weights are bounded, the ideal mismatch vector $Z$ always admits the desired decomposition into local correction sets (Lemma~\ref{lemma:LZdecoderworks}). These lemmas allow us prove our main result (Theorem~\ref{thm:main}): as the weight of $\tilde{Z}$ decreases throughout the steps of the algorithm, the residual error weight must also decrease. The analysis of the parallel decoder then builds upon this bound, with the additional requirement of showing that the decomposition of $Z$ into local corrections must be essentially disjoint (Lemma~\ref{lem:parallelmismatchreduction}).

In the remainder of this section, we set up notation and provide some preliminary results used in the proofs of Theorems~\ref{thm:main} and~\ref{thm:parallel_residual}. We first define quantities relating to the states of the decoders. Given the local structure of the quantum Tanner codes, it will be more convenient to bound the size of the syndrome noise in terms of its vertex support.

\begin{definition}\label{def:vertexsupport}
Given a quantum Tanner code and a syndrome noise $D$, let us define $D_v$ to be the restriction of $D$ to the set of stabilizer generators associated with vertex $v$. We define the vertex support of $D$ to be the set of all vertices such that $D_v \neq 0$. We denote the size of the vertex support by $|D|_V$. Note that we have $ \Delta^{-2}|D| \le r^{-1}|D| \le |D|_V \le |D|$, where $r$ is the number of stabilizer generators associated with the local code. 
\end{definition}

Given the noisy syndrome $\td \sigma = H_Ze + D$, let $\td \sigma_v$ denote the restriction of $\td \sigma$ to the checks associated with the vertex $v$. For each vertex $v\in V_1$, the decoder finds a locally minimal correction $\td\varepsilon_v$ such that $\sigma_v(\td\varepsilon_v) = \td\sigma_v$ . In the event that no local correction $\td\ep_v$ exists for $\td\sigma_v$, we may define $\td\sigma_v$ arbitrarily. In our case, we will simply define $\td\ep_v = 0$ by convention. If $\ep_v$ is the locally minimal correction associated with the noiseless syndrome $\sigma_v$, then we can decompose $\td\ep_v$ into ``noiseless'' and ``noisy'' parts as
\begin{equation}
	\td\ep_v = \ep_v + \ep_v(D)\, ,
\end{equation}
where $\ep_v(D)$ is defined by $\ep_v(D) = \td\ep_v - \ep_v$.  Note that $\ep_v(D)$ will be non-zero only when $D$ has non-zero support on $v$.

The full noisy mismatch vector initialized by the decoder is given by
\begin{align}
	\tilde Z &= \sum_{v\in V_1}\tilde\varepsilon_v = \sum_{v\in V_1}(\varepsilon_v+\varepsilon_v(D))\, .
\end{align}
It will likewise be convenient to split the mismatch into a noiseless and a noisy part, defined by
\begin{align}
Z = \sum_{v\in V_1}\varepsilon_v\quad \text{and}\quad  Z_N = \sum_{v\in V_1}\varepsilon_v(D)\, ,~\label{eq:noiseless_def}
\end{align}
so that $\tilde{Z} = Z + Z_N$. We will also need the restrictions of these vectors onto the vertices of $V_{01}$, which we define as
\begin{align}
    \td Z^{01} = \sum_{v \in V_{01}}\td \varepsilon_v,\quad Z^{01} = \sum_{v\in V_{01}}\varepsilon_v,\quad\text{and}\quad Z_N^{01} = \sum_{v\in V_{01}} \ep_v(D)\, .
\end{align}

The key idea of the proof is to pre-process the error using $\tilde{Z}^{01}$, and apply the local corrections $x_v$ step by step. Specifically, we define the initial pre-processed error $\tilde{e}_0$, and the ``noiseless'' pre-processed error $e_0$, by
\begin{align}
    \tilde{e}_0 &= e + \td Z^{01} = e + \sum_{v\in V_{01}}\tilde{\varepsilon}_v\, ,~\label{eq:pre-process}\\
    e_0 &= e + Z^{01} = \tilde{e}_0 + Z_N^{01}\, .
\end{align}
For the purpose of our proof, we consider the vector $\tilde{e}_0$ as the initial error state of the algorithm, and $\tilde{Z}_0 = \tilde{Z}$ as the initial mismatch. Note that in practice it does not matter at what point in the decoding procedure the set $\td Z^{01}$ is flipped. The pre-processing is only introduced as a convenience in our proof in order to relate the weight of $e$ to the weight of $Z$. The original algorithms considered in Ref.~\cite{leverrier2022parallel} involve a ``post-processing'' step instead, where $\tilde{Z}^{01}$ is applied at the very end rather than the beginning. Since the sets of qubits flipped are ultimately the same in either case, the results here hold without modification. 

The core loop of the decoding algorithm finds, at each step $i$, some local codeword $x_i = r_i + c_i \subseteq Q(v_i)$ such that
\begin{align}
    |\tilde{Z}_{i-1}| - |\tilde{Z}_{i-1} + x_i| \ge (1-\varepsilon)|x_i|.\label{ineq:decoding_condition}
\end{align}
Having found a codeword which satisfies~\eqref{ineq:decoding_condition}, we update the error and the mismatch vectors by
\begin{align}
	\tilde{e}_i = \tilde{e}_{i-1} + f_i,\quad\text{and}\quad \td Z_i = \td Z_{i-1} + x_i\, ,
\end{align}
where the flip-set $f_i \subseteq Q(v_i)$ is defined by
\begin{align}
    f_i=\begin{cases}
    0 & v_i \in V_{10}\ ,\\
    x_i & v_i \in V_{01}\ ,\\
    c_i & v_i \in V_{11}\ , \\
    r_i & v_i \in V_{00}\ .\\
    \end{cases}
\end{align}
Likewise, we can define the associated ``noiseless'' error and mismatch at each step by
\begin{align}
    e_i = e_{i-1} + f_i = \td{e}_i + Z_N^{01},\quad\text{and}\quad Z_i = Z_{i-1} + x_i = \tilde{Z}_i + Z_N\, .
\end{align}
Note that $Z_N$ and $Z_N^{01}$ are determined entirely by the syndrome noise $D$ and initial error $e$, and are constant through the decoding process. 

In the presence of measurement errors, it is no longer true that the noisy mismatch $\tilde{Z}$ can be decomposed into a sum of local codewords.\footnote{In the case of perfect syndrome measurements, we have $$Z = \sum_{v \in V_1}\varepsilon_v = \sum_{v \in V_1}(e_v+r_v+c_v) = \sum_{v\in V_1}(r_v+c_v),$$ 
where $r_v+c_v$ is the codeword that the local error is corrected to: $e_v+\varepsilon_v = r_v+c_v$. This decomposition no longer holds in the presence of imperfect measurements.} As such, some care must be taken in characterizing what exactly we mean by a ``mismatch''. This is captured by the definition below.

\begin{definition}\label{def:mismatch}
A mismatch vector is any $Z\in \bb F_2^Q$ that can be decomposed as $Z=C_0+C_1+R_0+R_1$, where 
\begin{align}
C_{j} = \sum_{v \in V_{\overline{j}j}} c_v\quad \text{and}\quad R_{i} = \sum_{v \in V_{i\overline{i}}} r_v
\end{align}
are the sum of local column codewords $c_v \in C_A \otimes \mathbb{F}_2^B$ and row codewords $r_v \in \mathbb{F}_2^A \otimes C_B$ on $Q(v)$, i.e., a mismatch vector is an element in the span of local codewords $C_1^\perp$. Here, we define $\overline i = 1-i$ for convenience.

The division of $Z$ into local codewords of the form $(C_0,C_1,R_0,R_1)$ is called a \emph{decomposition} of $Z$. Any given mismatch vector $Z$ may have many distinct decompositions. Given any decomposition, we define its \emph{weight} by
\begin{align}
\mathrm{wt}(C_0,C_1,R_0,R_1) = \|C_0\| + \|C_1\| + \|R_0\| + \|R_1\|,
\end{align}
where $\|C_i\|$ and $\|R_i\|$ denote the number of non-zero columns and rows, respectively, present in $C_i$ and $R_i$. Note that the weight is well-defined since distinct local codewords $c_v \subseteq C_i$ and $r_v\subseteq R_i$ are disjoint. We then define the \emph{norm} of a mismatch to be
\begin{align}
    \|Z\| = \min_{\substack{(C_0,C_1,R_0,R_1)\\Z = C_0 + C_1 + R_0 + R_1}}\mathrm{wt}(C_0, C_1,R_0,R_1).
\end{align}
Decompositions such that $\mathrm{wt}(C_0,C_1,R_0,R_1) = \|Z\|$ are called \emph{minimal weight} decompositions for $Z$.
\end{definition}

Note that technically the vector $\tilde{Z}$ which we call the noisy mismatch vector is \emph{not} a mismatch vector at all as defined by Definition~\ref{def:mismatch}. Nevertheless, we will continue to call $\tilde{Z}$ the noisy mismatch since there is little chance of confusion. The noiseless part $Z$ is a genuine mismatch vector by definition. The properties of the noiseless mismatch $Z$ are characterized by the following lemma from Ref.~\cite{leverrier2022parallel}.

\begin{lemma}[Lemma 17 in Ref.~\cite{leverrier2022parallel}]
	\label{lem:minimalZbound}
 Let $e\in \mathbb{F}_2^Q$ be an error and let $\varepsilon_v$ be a local minimal correction for $e_v$ at every vertex $v \in V_1$. Let
\begin{align}
    Z = \sum_{v\in V_1}\varepsilon_v\, .
\end{align}
Then $Z$ is a mismatch vector which satisfies
\begin{align}
|Z|\le 4|e|_R,\quad \text{and}\quad \|Z\|\le \frac{4}{\kappa\Delta}|e|_R\,.
\end{align}
\end{lemma}

The main purpose of pre-processing in our proof is that the noiseless pre-processed error $e_0$ and the noiseless mismatch $Z_0$ can be easily related through the following property.

\begin{definition}\label{def:weighted_error}
Let $e \in \mathbb{F}_2^Q$ be an error. We say that the error is $V_{ij}$-weighted if $\sigma_v(e) = 0$ for all $v \in V_{\overline{i}\overline{j}}$. Given a $V_{ij}$-weighted error $e$, we say that a mismatch vector $Z$ is \emph{associated} with $e$ if $\sigma_v(Z) = \sigma_v(e)$ for all $v \in V_{ij}$.
\end{definition}

\begin{lemma}
	\label{lemma:V10weightedmismatchZ0}
	The quantity $e_0$ is a $V_{10}$-weighted error and $Z_0 = Z$ is a mismatch vector associated with $e_0$.
\end{lemma}
\begin{proof}
First, we show that $Z$ is a mismatch vector. Note that $Z$ is the sum of local minimal corrections $\varepsilon_v$ to the error $e$, i.e., 
\begin{align}
    Z = \sum_{v\in V_1}\varepsilon_v\,,
\end{align} 
where for each vertex $v \in V_1$ we have $e_v = \varepsilon_v + x_v$ for some $x_v \in C_1^\perp$. Therefore
\begin{align}
Z = \sum_{v\in V_1}(e_v + x_v) = \sum_{v\in V_1}x_v\,,
\end{align}
where the $e_v$ terms cancel since each face occurs exactly twice in the sum above. Next, we show that $e_0$ is $V_{10}$-weighted. We have
\begin{align}
     e_0 = e + Z^{01} = e + \sum_{v\in V_{01}}\varepsilon_v\,.
\end{align}
Note that the terms in the latter sum are disjoint for distinct vertices $v,v'\in V_{01}$. It follows that the restriction of $e_0$ to a vertex $v \in V_{01}$ is given by
\begin{align}
    (e_0)_v = e_v + \varepsilon_v = x_v\,,
\end{align}
which has zero syndrome. Finally we show that $Z$ is associated with $e_0$. The restriction of $e_0$ to a vertex $v\in V_{10}$ is given by
 \begin{align}
    (e_0)_v = e_v + Q(v)\cap \sum_{u \in V_{01}} \varepsilon_u\,.
\end{align}
Likewise, the restriction of $Z$ to $v\in V_{10}$ is given by
\begin{align}
Z_v = \varepsilon_v + Q(v)\cap \sum_{u \in V_{01}} \varepsilon_u\,.
\end{align}
It follows that
\begin{align}
    \sigma_v(Z) = \sigma_v(\varepsilon_v) + \sigma_v\left(Q(v)\cap \sum_{u \in V_{01}} \varepsilon_u\right) = \sigma_v(e_v) + \sigma_v\left(Q(v)\cap \sum_{u \in V_{01}} \varepsilon_u\right) = \sigma_v\left(e_0\right),
\end{align}
which shows that $Z$ is associated with $e_0$.
\end{proof}

The notion of $e_i$ being a $V_{10}$-weighted error is invariant as the decoder proceeds, i.e., if $e_i$ is initially $V_{10}$-weighted then it remains so. Moreover, if $Z$ was initially a mismatch associated with $e_0$ then $Z_i$ remains associated with $e_i$ throughout all steps $i$ of the decoder.

\begin{lemma}
	\label{lemma:V10weightedmismatchinduction}
	Let $Z$ be a weighted mismatch vector associated with a $V_{10}$-weighted error $e$. Let $x=c+r\subseteq Q(v)$ be a codeword of $C_1^\perp$, with $v\in V_{ij}$. Define
	\begin{equation}
		f=\begin{cases}
			0, & v \in V_{10}\, ,\\
			x, & v \in V_{01}\, ,\\
			c, & v \in V_{11}\, , \\
			r, & v \in V_{00}\, ,
		\end{cases}
	\end{equation}
	to be the associated flip set. Then $e+f$ is again a $V_{10}$-weighted error and $Z+x$ is an associated mismatch vector.
\end{lemma}
\begin{proof}
	It is clear that $Z+x$ is a mismatch vector since $Z$ was one and we add a single $C_1^\perp$ codeword.

    We first show that $e+f$ remains $V_{10}$-weighted. Clearly $e+f$ is $V_{10}$-weighted if $v \in V_{10}$ or $v\in V_{01}$ since we either add nothing, or a local codeword to a $V_{01}$ vertex. Now suppose that $v \in V_{00}$ so that $f=r$. We can decompose $r$ into $r = r_1 + \cdots + r_k$, where each $r_i$ is a local codeword supported on a single row, which we can assume to be indexed by the edge $(v,u_i)$ for some $u_i \in V_{01}$. The syndrome of $e+r$ on a vertex $u\in V_{01}$ is therefore given by
    \begin{align}
    \sigma_u(e+f) =  \begin{cases} \sigma_u(e)& u \neq u_i \text{ for all }i,\\
                                   \sigma_u(e + r_i) & u = u_i \text{ for some } i.
                    \end{cases}
    \end{align}
    In either case, we have $\sigma_u(e+f)=0$ so that $e+f$ is $V_{10}$-weighted. The case where $v\in V_{11}$ is analogous, taking $f=c$ and making a similar decomposition.
    
	Finally, we show that $Z+x$ is associated with $e+f$. Let us write
    \begin{align}
    Z = \sum_{u\in V_{10}}\varepsilon_u,
    \end{align}
    where $\sigma_u(Z) = \sigma_u(e)$ for all $u \in V_{10}$. If $v \in V_{10}$ then there is nothing to show since all syndromes are unchanged. If $v \in V_{01}$ then define
	\begin{equation}
		\ep_u' = \ep_u + Q(u)\cap x\, 
	\end{equation}
    so that
    \begin{align}
        Z + x = \sum_{u \in V_{10}}\ep_u'.
    \end{align}
	Since $(e+x)_u = e_u + Q(u)\cap x$, we see that $\ep_u'$ has the same syndrome as $(e+f)_u$.

    Lastly, suppose $v \in V_{00}$, with the $V_{11}$ case being analogous. Let $f = r$. Note that $\varepsilon_u' = \varepsilon_u$ and $(e+r)_u = e_u$ for all $u\in V_{10}$ not adjacent to $v$. Therefore it suffices to consider $u \in N(v)$. In this case, $Q(u)\cap c$ is just the column of $c$ labeled by the edge $(u,v)$ and so $Q(u)\cap c$ is a local codeword. Therefore $\sigma_u(c) = 0$. It follows that
    \begin{align}
        \sigma_u(\varepsilon_u') = \sigma_u(\varepsilon_u) + \sigma_u(x) = \sigma_u(e) + \sigma_u(r) + \sigma_u(c) = \sigma_u(e) + \sigma_u(r) = \sigma_u(e+r)
    \end{align}
    for all $u \in N(v)$. Therefore $\sigma_u(Z+x) = \sigma_u(e+f)$ for all $u\in V_{10}$ and so $Z+x$ is associated with $e+f$.
\end{proof}

Lemma~\ref{lemma:V10weightedmismatchZ0} and Lemma~\ref{lemma:V10weightedmismatchinduction} show that $Z_i$ is a mismatch vector associated with the $V_{10}$-weighted error $e_i$ for all $i$. We further cite the following lemma from Ref.~\cite{leverrier2022parallel}, which gives a sufficient condition for the existence of good local corrections. This is the key to proving that in the noiseless case, the sequential and parallel decoders converge. 

\begin{definition}
Let $Z$ be a mismatch vector and let $Z = C_0 + C_1 + R_0 + R_1$ be a minimal decomposition for $Z$. We say that a vertex $v\in V_{ij}$ is \emph{active} with respect this decomposition if $Q(v)\cap (R_i + C_j) \neq 0$. 
\end{definition}

\begin{theorem}[Theorem 12 in Ref.~\cite{leverrier2022parallel}]
	\label{thm:LZflipexists}
	Fix $\delta\in (0,1)$. Let $Z$ be a non-zero mismatch vector. If for all $i,j\in \{0,1\}$, the set of active vertices $S_{ij}\subseteq V_{ij}$ for a minimal decomposition of $Z$ satisfies
	\begin{equation}
		|S_{ij}| \le \frac{1}{2^{12}}d_r^2\delta^3\kappa |V_{00}|\, ,
	\end{equation}
	where $d_r$ denotes the relative distance of the local code, then there exists a non-zero $x\subseteq Q(v)$ for some $v\in V_{ij}$ that is a $C_1^\perp$ codeword such that
	\begin{equation}
		|Z| - |Z+x| \ge (1-\delta)|x|\, .
	\end{equation}
\end{theorem}

\subsection{Sequential Decoder}

To begin analyzing the sequential decoder with noisy input, the natural question to ask is that if the ideal mismatch $Z$ can be decomposed by Algorithm~\ref{alg:seq_mismatch} into  $\mathcal{F} = \{x_i\}_{i=1}^t$, how well do these local corrections $x_i$ decompose the noisy mismatch $\tilde{Z} = Z + Z_N$? The following two lemmas address this question.

\begin{definition}\label{def:decomposability}
Let $Z$ be a mismatch vector. We say that $Z$ is \emph{$\delta$-decomposable} if Algorithm~\ref{alg:seq_mismatch} successfully returns a decomposition of $Z$ when run with parameter $\delta$, i.e., if Algorithm~\ref{alg:seq_mismatch} halts with state $\hat{Z} = 0$.
\end{definition}

\begin{lemma}\label{lem:decomp}
Let $Z$ be an $\delta$-decomposable mismatch and let $\mathcal{F} = \{x_i\}_{i=1}^t$ denote the codewords returned by Algorithm~\ref{alg:seq_mismatch} run with input $Z$ and parameter $\delta$. Then
\begin{align}
        (1-\delta) \sum_{i=1}^t |x_i| \le |Z| \le \sum_{i=1}^t|x_i|\,.
\end{align}
\end{lemma}
\begin{proof}
    Let 
    \begin{align}
    Z_k = Z - \sum_{i=1}^k x_i\,,
    \end{align}
    with $Z = Z_0$. Note that since Algorithm~\ref{alg:seq_mismatch} completely decomposes $Z$, we have $Z_t = 0$ and
    \begin{align}
        Z = \sum_{i=1}^t x_i\,.
    \end{align}
    For the decomposition with parameter $\delta$, we have $|Z_{i-1}| - |Z_i| \ge (1-\delta)|x_i|$ and therefore
	\begin{equation}
		|Z| \ge (1-\delta) \sum_{i=1}^t |x_i|\,.
	\end{equation}
    Together, we get the bounds
    \begin{align}
        (1-\delta) \sum_{i=1}^t |x_i| \le |Z| \le \sum_{i=1}^t|x_i|\,.
    \end{align}
\end{proof}

\begin{lemma}
	\label{lem:flipexists}
	Let $Z$ be a mismatch vector and let $Z_N\in \bb F_2^Q$ be any vector. Let $\td Z = Z + Z_N$. Suppose that $Z$ is $\delta$-decomposable with decomposition $\mc F = \{x_i\}_{i=1}^t$. Let
    \begin{equation}
        \mc F^* = \{x\in \mc F: |\td Z|-|\td Z + x|\ge (1-\varepsilon) |x|\}\, .
    \end{equation}
    Then
	\begin{equation}
		\sum_{x\in \mc F^*} |x| \ge c_1|Z| - c_2|Z_N| \label{ineq:Z_N_bound}
	\end{equation}
	for constants
 \begin{align}
     c_1 = \frac{\varepsilon-2\delta}{\varepsilon(1-\delta)}\quad\text{and}\quad c_2 = \frac{2}{\varepsilon}.  
 \end{align}
 In particular, if $\mc F^*=\emptyset$, then $c_1|Z|\le c_2|Z_N|$.
\end{lemma}

\begin{proof}
	This proof follows the idea of Lemma~5.1 in Ref.~\cite{grospellier2019thesis}. 
	Given any set $y \in \mathbb{F}_2^Q$, we have
	\begin{align}
		|\tilde{Z}| - |\tilde{Z} + y| &= |\tilde{Z}| - (|\tilde{Z}|+ |y| - 2|\tilde{Z} \cap y|) = 2|\tilde{Z} \cap y| - |y|\, .
	\end{align}
	For all $y \in \mathcal{F}\setminus\mathcal{F}^*$, we have
	\begin{equation}
		|\tilde{Z} \cap y| = \frac{1}{2}(|y| + |\tilde{Z}| - |\tilde{Z} + y|) < \left(1-\frac{\varepsilon}{2}\right)|y|\, .
	\end{equation}
	Define $T = \sum_{x\in \mathcal{F}}|\tilde{Z}\cap x|$. We then have
	\begin{align}
		T
		&= \sum_{x \in \mathcal{F}^*}|\tilde{Z}\cap x| + \sum_{y\in \mathcal{F}\setminus\mathcal{F}^*}|\tilde{Z}\cap y| \\
		&< \sum_{x\in \mathcal{F}^*}|x| + \left(1-\frac{\varepsilon}{2}\right)\sum_{y\in \mathcal{F}\setminus\mathcal{F}^*}|y| \\
		&= \frac{\varepsilon}{2}\sum_{x\in \mathcal{F}^*}|x| + \left(1-\frac{\varepsilon}{2}\right)\sum_{y\in \mathcal{F}}|y| \\
		&\le \frac{\varepsilon}{2}\sum_{x\in \mathcal{F}^*}|x| + \frac{2-\varepsilon}{2(1-\delta)}|Z|\, ,
	\end{align}
    where the last inequality follows from Lemma~\ref{lem:decomp}. On the other hand, we also have
    \begin{align}
        T &\ge |\tilde{Z}\cap \sum_{x\in\mathcal{F}} x| = |\tilde{Z}\cap Z| \\
		&= |Z| - |Z\cap Z_N| \ge |Z| - |Z_N|\, .
    \end{align}
	Combining these two inequalities, we get
    \begin{align}
        \frac{\varepsilon}{2}\sum_{x\in \mathcal{F}^*}|x| + \frac{2-\varepsilon}{2(1-\delta)}|Z| \ge |Z| - |Z_N|,
    \end{align}
    or equivalently
    \begin{align}
       \sum_{x\in \mathcal{F}^*}|x| \ge \frac{\varepsilon - 2\delta}{\varepsilon(1-\delta)}|Z| - \frac{2}{\varepsilon}|Z_N|,
    \end{align}
    as desired.
\end{proof}

Note that Lemma~\ref{lem:flipexists} will set an implicit bound of $\delta < 1/2$ since we require $\varepsilon - 2\delta > 0$ for the bound~\eqref{ineq:Z_N_bound} to be non-trivial.

Suppose now that the noisy mismatch vector $\tilde{Z}$ is given as input to Algorithm~\ref{alg:seq_decoder} with parameter $\ep$, which terminates after $T$ iterations. Let us denote the residual error by $\tilde{e}_T$ and its associated mismatch by $\tilde{Z}_T = Z_T + Z_N$. If $Z_T$ is $\delta$-decomposable, then Lemma~\ref{lem:flipexists} implies that $|Z_T| = O(Z_N)$. Namely, the sequential decoder terminates only when the mismatch noise $Z_N$ becomes significant. In the following lemma, we further relate the weight of the noiseless residual error $e_T$ with $|Z_T|$.

\begin{lemma}[Mismatch Correctness and Soundness]
    \label{lem:Zsoundness}
    Let $e$ be a $V_{10}$-weighted error and let $Z$ be an associated mismatch vector. Suppose that $Z$ is $\delta$-decomposable and that
    \begin{align}~\label{eqn:cond-Zsoundness}
        |e|_R + \frac{1}{\kappa(1-\delta)}|Z| < d\, .
    \end{align}
    Then we have
    \begin{align}
    |Z|\ge (1-\delta)\kappa |e|_R\, .
    \end{align}
\end{lemma}
\begin{proof}
    Let $\mc F = \{x_i\}_{i=1}^t$ denote the decomposition returned for $Z$ by Algorithm~\ref{alg:seq_mismatch} with parameter $\delta$. Each $x_i$ is supported on the local view of some vertex $v_i$ and has the further decomposition into column and row codewords as $x_i=c_i+r_i$.

    First, we prove that $ e \cong \hat{C}_1 + \hat{R}_0$, where $\cong$ denotes equivalence up to stabilizers. Let $e_0 = e$ and define $ e_{i} =  e_{i-1} + f_{i}$ where
    \begin{equation}
		f_i=\begin{cases}
			0, & v \in V_{10}\, ,\\
			x_i, & v \in V_{01}\, ,\\
			c_i, & v \in V_{11}\, , \\
			r_i, & v \in V_{00}\, .\\
		\end{cases}
    \end{equation}
    Note that by construction we have
    \begin{align}
        e_t = e_0 + \hat{C}_1 + \hat{R}_0.
    \end{align}
    By Lemma~\ref{lemma:V10weightedmismatchinduction}, the errors $e_i$ are all $V_{10}$-weighted, and the vector $Z_k = Z + \sum_{i=1}^kx_i$ is a mismatch vector associated with $e_i$ at each step. It follows by the $V_{10}$-weighting of $e_t$ that
    \begin{align}
        \forall v\in V_{01}:\quad \sigma_v(e_t) = 0.
    \end{align}
    Since $Z_t = 0$, it follows by the association of $Z_t$ and $e_t$ that 
    \begin{align}
        \forall v\in V_{10}:\quad \sigma_v(e_t) = \sigma_v(Z_t) = 0.
    \end{align}
    It follows that $e_t$ has zero syndrome. It remains to show that $e_t$ is a stabilizer, which we can do by bounding its weight. For each flip-set $f_i$, we have
	\begin{align}
		|f_i| \le |r_i|+|c_i| \le \Delta (\|r_i\|+\|c_i\|) \le |x_i|/\kappa\, ,\label{eq:flip_bound}
	\end{align}
    where we use the robustness of the local code in the last inequality. Using Lemma~\ref{lem:decomp}, we then have
	\begin{align}
		|Z| \ge (1-\delta)\sum_{i=1}^t|x_i| \ge (1-\delta)\kappa \sum_{i=1}^t|f_i|.
	\end{align}
    It follows that
    \begin{align}
        |e_t|_R = \left|e + \sum_{i=1}^tf_i\right|_R \le |e|_R + \sum_{i=1}^t|f_i| \le |e|_R + \frac{1}{\kappa(1-\delta)}|Z| < d.
    \end{align}
    Therefore $e_t \cong 0$ and hence $e \cong \hat{C}_1 + \hat{R}_0$. Finally, we have
    \begin{align}
        |e|_R = \left|\hat{C}_1 + \hat{R}_0\right|_R \le \left|\hat{C}_1 + \hat{R}_0\right| \le \sum_{i=1}^t|f_i| \le \frac{1}{\kappa(1-\delta)}|Z|. 
    \end{align}
\end{proof}

Now we show that, without surprise, $Z_T$ is $\delta$-decomposable. Let us define the constants
\begin{align}
    A_{\varepsilon} &= \frac{24}{\kappa\Delta(1-\varepsilon)},\quad B_{\varepsilon} = \frac{3\Delta}{\kappa(1-\varepsilon)},\quad \text{and}\quad C_\delta = \frac{1}{2^{12}}d_r^2\delta^3\kappa\Delta^{-2}\, .
\end{align}

For the purposes of the parallel decoder, it will be convenient to consider a generalized mismatch decomposition procedure which initially starts the decomposition with some weight parameter $\varepsilon$ and then switches to some other parameter $\varepsilon'$ partway through (see Lemma~\ref{lem:parallelmismatchreduction}). We state the generalized result below in Lemma~\ref{lemma:LZdecoderworks}, although we will only need the special case where $\varepsilon = \varepsilon'$ for the analysis of the sequential decoder.

\begin{lemma}
	\label{lemma:LZdecoderworks}
    Let $e$ be an error and $D$ a syndrome noise. Let $\tilde{Z} \equiv Z + Z_N$ denote the initial noisy mismatch vector assigned to $e$ and $D$. 
    
    Let $\varepsilon,\varepsilon' \in (0,1)$ be constants such that $\varepsilon' \le \varepsilon$. Consider a modified Algorithm~\ref{alg:seq_mismatch} which takes input $\tilde{Z}$ and runs with parameter $\varepsilon$ for the first $t$ steps and then switches to parameter $\varepsilon'$ until it halts at step $T \ge t$. Let $\tilde{Z}_T \equiv Z_T + Z_N$ denote the final output of this process.
    
    If $A_{\varepsilon}|e|_R+B_{\varepsilon}|D|_V\le C_\delta n$, then $Z_T$ is $\delta$-decomposable.
\end{lemma}

\begin{proof} 
Consider the process of running the modified Algorithm~\ref{alg:seq_mismatch} with input $\tilde{Z}$ and parameter $\varepsilon$ for $t$ steps, and then switching the parameter to $\varepsilon'$ until the algorithm finally halts at step $T$. Let $\{x_1,\dots,x_t\}$ be local codewords obtained with parameter $\varepsilon$, and $\{x_{t+1}, \dots, x_T\}$ the codewords obtained with parameter $\varepsilon'$. Denoting $\td Z_i$ the mismatch vector at iteration $i$, we have
\begin{align}
|\td Z_{i-1}| - |\td Z_i| \geq \begin{cases}
    (1-\varepsilon) |x_i|, & i \in \{1, \dots, t\}\, , \\
    (1-\varepsilon') |x_i|, & i \in \{t+1, \dots, T\} \ . \label{eq:alphatobeta}
\end{cases}
\end{align}
We wish to show that $Z_T$ is $\delta$-decomposable. Suppose that Algorithm~\ref{alg:seq_mismatch} returns local codewords $\{y_1, \dots, y_K\}$ when given input $Z_T$ with parameter $\delta$. Let $S_{T+k,ij}$ denote a set of active vertices in $V_{ij}$ for the mismatch 
\begin{align}
Z_{T+k} \equiv Z_T + \sum_{\ell=1}^k y_\ell\, .
\end{align}
For all $k \in [K]$, we have
\begin{align}
    |S_{T+k,ij}| &\le \|Z_{T+k}\| \le  \|Z\| + \sum_{i=1}^T\|x_i\| + \sum_{\ell=1}^k\|y_\ell\|\, ,\label{ineq:1}
\end{align}
where the first inequality holds since there exists at least one non-zero row or column for each active vertex. By robustness of the local code, we have $\kappa\Delta\|x_i\| \le |x_i|$. Continuing the chain of inequalities, we have
\begin{align}
    \eqref{ineq:1} &\le \|Z\| + \frac{1}{\kappa\Delta}\sum_{i=1}^T|x_i| + \frac{1}{\kappa\Delta}\sum_{\ell=1}^k|y_\ell|\\
    &\le \|Z\| + \frac{1}{\kappa\Delta}\sum_{i=1}^T|x_i| + \frac{1}{\kappa\Delta(1-\delta)}|Z_T|\, ,\label{ineq:2}
\end{align}
where the first inequality follows by robustness and the second by the fact that $|Z_{T+\ell-1}| - |Z_{T+\ell-1} + y_\ell| \ge (1-\delta)|y_\ell|$. Using inequality~\eqref{eq:alphatobeta}, we get
    \begin{align}
        |Z_T| = \left|Z + \sum_{i=1}^T x_i\right| \le |Z| + \sum_{i=1}^T|x_i| \le |Z| + \frac{1}{1-\varepsilon}\left(|\tilde{Z}|-|\tilde{Z}_t|\right) + \frac{1}{1-\varepsilon'}\left(|\tilde{Z}_t|-|\tilde{Z}_T|\right).
    \end{align}
Since $\varepsilon' \le \varepsilon$, it follows that
\begin{align}
    |Z_T| \le |Z| + \frac{1}{1-\varepsilon}|\tilde{Z}|.\label{ineq:Z_T}
\end{align}
Inserting \eqref{ineq:Z_T} into \eqref{ineq:2}, we get
 \begin{align}
     |S_{T+k,ij}| &\le \|Z\| + \frac{1}{\kappa\Delta(1-\delta)}|Z| + \frac{1}{\kappa\Delta(1-\varepsilon)}\left(\frac{2-\delta}{1-\delta}\right)|\td Z|\label{ineq:4}\\
		&\le  \|Z\| + \frac{1}{\kappa\Delta(1-\delta)}|Z| + \frac{1}{\kappa\Delta(1-\varepsilon)}\left(\frac{2-\delta}{1-\delta}\right)(|Z|+|Z_N|)\\
		&\le \frac{4}{\kappa\Delta}|e|_R + \frac{4}{\kappa\Delta(1-\delta)}|e|_R + \frac{1}{\kappa\Delta(1-\varepsilon)}\left(\frac{2-\delta}{1-\delta}\right)(4|e|_R+\Delta^2|D|_V)\, ,\label{ineq:5}
 \end{align}
where the last inequality follows by applying Lemma~\ref{lem:minimalZbound}, together with the fact that $\varepsilon_v(D)$ can be non-zero only when $v$ is in the support of $D$ and hence
\begin{align}\label{eq:boundonZN}
    |Z_N| = \left|\sum_{v\in V_1}\varepsilon_v(D) \right| \le \sum_{v\in V_1}\left|\varepsilon_v(D) \right| \le |D|_V\max_{v \in V_1}|\varepsilon_v(D)| \le |D|_V\Delta^2.
\end{align}
Simplifying, we finally get
\begin{align}
    |S_{T+k,ij}| &\le \frac{4}{\kappa\Delta}\left(\frac{2-\delta}{1-\delta}\right)\left(\frac{2-\varepsilon}{1-\varepsilon}\right)|e|_R + \frac{\Delta}{\kappa(1-\varepsilon)}\left(\frac{2-\delta}{1-\delta}\right)|D|_V \\
    &\le \frac{12}{\kappa\Delta}\left(\frac{2}{1-\varepsilon}\right)|e|_R + \frac{3\Delta}{\kappa(1-\varepsilon)}|D|_V\\ 
    &\equiv A_{\varepsilon}|e|_R + B_{\varepsilon}|D|_V,
\end{align}
where we use the fact that $(2-\delta)/(1-\delta) \le 3$ for $\delta \in (0,1/2)$. It follows that if we have $A_{\varepsilon}|e|_R + B_{\varepsilon}|D|_V \le C_\delta n$, then the active vertex condition of Theorem~\ref{thm:LZflipexists} is always satisfied so that Algorithm~\ref{alg:seq_mismatch} must be able to completely decompose $Z_T$.
\end{proof}

It remains for us to check that~\eqref{eqn:cond-Zsoundness} in Lemma~\ref{lem:Zsoundness} holds.
\begin{lemma}\label{lem:Z_Tsoundness}
Assume the hypotheses of Lemma~\ref{lemma:LZdecoderworks}, and furthermore that 
\begin{align}
A_{\varepsilon}|e|_R+B_{\varepsilon}|D|_V \le \frac{d}{\Delta}.
\end{align}
Then
\begin{align}
    |Z_T|\ge (1-\delta)\kappa |e_T|_R.
    \end{align}
\end{lemma}
\begin{proof}
By Lemmas~\ref{lemma:V10weightedmismatchZ0} and~\ref{lemma:V10weightedmismatchinduction}, the error $e_T$ is $V_{10}$-weighted and $Z_T$ is an associated mismatch vector. Applying Lemma~\ref{lem:Zsoundness}, it suffices to prove
\begin{align}
        |e_T|_R + \frac{1}{\kappa(1-\delta)}|Z_T| < d\ .\label{ineq:distance_bound1}
\end{align}
We have
\begin{align}
    |e_T|_R &= \left|e_0 + \sum_{i=1}^Tf_i\right|_R \le  |e_0|_R + \sum_{i=1}^T|f_i|\le |e_0|_R + \frac{1}{\kappa}\sum_{i=1}^T|x_i| \le |e_0|_R + \frac{1}{\kappa(1-\varepsilon)}|\tilde{Z}|,
\end{align}
where the second inequality follows from~\eqref{eq:flip_bound}. We then get
\begin{align}
    |e_T|_R + \frac{1}{\kappa(1-\delta)}|Z_T| &\le |e_0|_R + \frac{1}{\kappa(1-\varepsilon)}|\tilde{Z}| + \frac{1}{\kappa(1-\delta)}|Z_T|\\
    &\le |e_0|_R + \frac{1}{\kappa(1-\varepsilon)}|\tilde{Z}| + \frac{1}{\kappa(1-\delta)}\left(|Z| + \frac{1}{1-\varepsilon}|\tilde{Z}|\right)\\
    &=|e_0|_R + \frac{1}{\kappa(1-\delta)}|Z| + \frac{1}{\kappa(1-\varepsilon)}\left(1+\frac{1}{1-\delta}\right)|\tilde{Z}|\, ,
\end{align}
where we use~\eqref{ineq:Z_T} in the second inequality. Next, we may assume without loss of generality that $e$ is a reduced error. Then we have
\begin{equation}
		|e_0|_R = |\tilde{e}_0 + Z_N^{01}|_R = \big|e + \sum_{v\in V_{01}} \ep_v\big|_R \le |e| + \sum_{v\in V_{01}}|e_v| = 2|e| = 2|e|_R\, ,
\end{equation}
where we use the fact that $\varepsilon_v$ are minimum weight corrections in the inequality above and the fact that $e_u\cap e_{v} = \emptyset$ for distinct vertices $u, v \in V_{01}$ in the second last equality. Following the same steps as from \eqref{ineq:4} to \eqref{ineq:5}, we therefore get
\begin{align}
    |e_T|_R + \frac{1}{\kappa(1-\delta)}|Z_T| &\le 2|e|_R + \frac{4}{\kappa(1-\delta)}|e|_R + \frac{1}{\kappa(1-\varepsilon)}\left(1+\frac{1}{1-\delta}\right)(4|e|_R+\Delta^2|D|_V)\\
    &\le \left[2 + \frac{4}{\kappa(1-\varepsilon)}\left(1 + \frac{2-\varepsilon}{1-\delta}\right)\right]|e|_R + \Delta B_{\varepsilon} |D|_V.
\end{align}
We can simplify the inequality above by noting that $\kappa \le d_r \le 1$~\cite{KPTwoSided}. Then we have
\begin{align}
    2 + \frac{4}{\kappa(1-\varepsilon)}\left(1 + \frac{2-\varepsilon}{1-\delta}\right) &= \frac{1}{\kappa}\left[2\kappa + \frac{4}{1-\varepsilon}\left(1 + \frac{2-\varepsilon}{1-\delta}\right)\right]\\
    &\le \frac{1}{\kappa}\left[4 + \frac{4}{1-\varepsilon}\left(1 + \frac{2-\varepsilon}{1-\delta}\right)\right]\\
    &= \frac{4}{\kappa}\left(\frac{2-\delta}{1-\delta}\right)\left(\frac{2-\varepsilon}{1-\varepsilon}\right)\\
    &\le \frac{24}{\kappa(1-\varepsilon)} =\Delta A_{\varepsilon}.
\end{align}
Therefore it suffices to require
\begin{align}
A_{\varepsilon}|e|_R + B_{\varepsilon}|D|_V \le \frac{d}{\Delta}
\end{align}
in order that inequality~\eqref{ineq:distance_bound1} holds.
\end{proof}

Combining the inequalities, we obtain the main result for sequential decoder.
\begin{theorem}[Main Theorem for the Sequential Decoder]
	\label{thm:main}
	Let $e$ be an error and let $D$ be a syndrome error. Suppose that
    \begin{align}
A_{\varepsilon}|e|_R+B_{\varepsilon}|D|_V &\le \min\left(C_\delta n,d/\Delta\right)\, .
    \end{align}
    Let $\tilde{\sigma} = \sigma(e) + D$. Then Algorithm~\ref{alg:seq_decoder} with input $\tilde{\sigma}$ and parameter $\varepsilon$ will output a correction $\hat{f}$ satisfying 
    \begin{align}
    |e+\hat{f}|_R\le \left(1+\frac{2c_2}{\kappa c_1}\right)\Delta^2|D|_V.
    \end{align}
\end{theorem}

\begin{proof}
	Suppose that Algorithm~\ref{alg:seq_decoder} with parameter $\varepsilon$ terminates after $T$ steps with output $\hat{f}$. Let $Z_T$ denote the state of the mismatch after the algorithm terminates. By Lemma~\ref{lemma:LZdecoderworks}, $Z_T$ is $\delta$-decomposable. This allows us to apply Lemma~\ref{lem:flipexists}, giving
	\begin{equation}
		0 \ge c_1|Z_T| - c_2|Z_N|\, ,\label{ineq:6}
	\end{equation}
    since the set $\mathcal{F}^*$ must be empty when Algorithm~\ref{alg:seq_decoder} with parameter $\varepsilon$ terminates. By Lemma~\ref{lem:Z_Tsoundness}, we get
	\begin{equation}
		|Z_T|\ge (1-\delta)\kappa |e_T|_R\, .\label{ineq:7}
	\end{equation}
	But we know
	\begin{align}
		|e_T|_R &= |\tilde{e}_T + Z_N^{01}|_R \ge |\tilde{e}_T|_R - |Z_N^{01}|_R \ge |\tilde{e}_T|_R - \Delta^2|D|_V\, .\label{ineq:8}
	\end{align}
	Combining the inequalities~\eqref{ineq:6}, \eqref{ineq:7}, and \eqref{ineq:8} finally gives
	\begin{align}
		|e + \hat{f}|_R &= |\tilde{e}_T|_R\\ 
        &\le |e_T|_R + \Delta^2|D|_V\\
		&\le \frac{1}{(1-\delta)\kappa}|Z_T| + \Delta^2|D|_V\\
		&\le \frac{c_2}{c_1(1-\delta)\kappa}|Z_N| + \Delta^2|D|_V\\
		&\le \left(1+\frac{c_2}{c_1(1-\delta)\kappa}\right)\Delta^2|D|_V\, . \label{eq:seqresidualerrorbound}
	\end{align}
Note that the restriction $\delta < 1/2$, as required by Lemma~\ref{lem:flipexists}, implies that $(1-\delta)^{-1} \le 2$.
\end{proof}

This completes our proof of the main theorem for the sequential decoder.

\subsection{Parallel Decoder}

The key idea in analyzing the parallel decoder is to compare the performance of one iteration of parallel decoding to that of a full execution of the sequential decoder. Our convention in this section will be that superscript indices will denote the parallel decoding iteration (always with parameter $1/2$), while subscript indices will denote the sequential decoding iteration. For example, $\tilde{Z}^{(k)}_j$ denotes the mismatch obtained after $k$ iterations of parallel decoding and then $j$ iterations of sequential decoding.

For convenience, we will fix some parameters in this section. Throughout, we will take $\varepsilon = 1/2$ for the parallel decoder. We will write $A = A_{\varepsilon=1/2}$ and $B = B_{\varepsilon = 1/2}$.

\begin{lemma}\label{lem:parallelmismatchreduction}
Let $\varepsilon' \in (0, 1/6)$. Let $\tilde{Z}^{(k)}$ denote the current state of the (noisy) mismatch vector. Let $\tilde{Z}_T^{(k)}$ denote the residual mismatch after running the sequential decoder with input $\tilde{Z}^{(k)}$ and parameter $\varepsilon'$. Then after one iteration of parallel decoding, the weight of the mismatch is reduced by at least
\begin{align}
    |\tilde{Z}^{(k)}| - |\tilde{Z}^{(k+1)}| \ge \frac{1}{16}(1-6\varepsilon')\left(|\tilde{Z}^{(k)}| - |\tilde{Z}_T^{(k)}|\right).
\end{align}
\end{lemma}
\begin{proof}

The proof closely follows the ideas of Lemma 18 in Ref.~\cite{leverrier2022parallel}. For ease of notation we write $\tilde{Z}^{(k)}$ as $\tilde{Z}$ throughout this proof. Suppose that Algorithm~\ref{alg:seq_mismatch} runs with input $\tilde{Z}$ and parameter $\varepsilon'$ returns local codewords $\{x_i\}_{i=1}^T$ and residual mismatch $\tilde{Z}_T$. Therefore we can write
\begin{align}
    \tilde{Z} = \sum_{i=1}^T x_i + \tilde{Z}_T\,.
\end{align}

We will analyze the overlap among the sets $x_i$, and argue that the parallel decoder's output will intersect non-trivially with the sequential decoder's output. Let us define the sets
\begin{align}
    x_i' = \left(\tilde{Z}\cap x_i\right)\setminus \bigcup_{j<i}x_j\,.
\end{align}
Note that the sets $x_i'$ are disjoint, and that they satisfy
\begin{align}
    \bigcup_{i=1}^T x_i' = \tilde{Z} \cap \bigcup_{i=1}^T x_i \supseteq \tilde{Z} \cap \sum_{i=1}^Tx_i\, ,
\end{align}
which implies
\begin{align}
    \left|\tilde{Z}\setminus \bigcup_{i=1}^Tx_i'\right| \le \left|\tilde{Z}\setminus \left(\tilde{Z} \cap \sum_{i=1}^T x_i\right)\right|=\left|\tilde{Z}\setminus \sum_{i=1}^T x_i\right|\le \left|\tilde{Z} + \sum_{i=1}^Tx_i\right| = |\tilde{Z}_T|\,.
\end{align}
Next, we define the set of ``good'' indices $G \subseteq [T]$ such that $i \in G$ if and only if
\begin{align}
    |x_i'| \ge \left(1 - \frac{3}{2}\varepsilon'\right)|x_i|\,.
\end{align}
Let $B = [T]\setminus G$ denote the remaining set of ``bad'' indices. For each $j \in [T]$, let us define
\begin{align}
    \td Z'_j = \tilde{Z}\setminus \bigcup_{i\le j} x_i' = \td Z'_{j-1}\setminus x_j'\,.
\end{align}
We wish to bound the difference between $\tilde{Z}_j$ and $\tilde{Z}_j'$. Let us denote this difference by
\begin{align}
A_j = \tilde{Z}_j\setminus \td Z_j' \,. 
\end{align}
To bound the size of $A_j$, we examine how the size of $\tilde{Z}$ changes as we update it by adding codewords $x_j$. Since the $x_j$'s were obtained by running the decoder with parameter $\varepsilon'$, it follows that
\begin{align}
    |\tilde{Z}_{j-1}\cap x_j| \ge (1-\varepsilon'/2)|x_j|\,.
\end{align}

\begin{figure}
{\centering

\begin{tikzpicture}[x=0.9pt,y=0.9pt,yscale=-1,xscale=1]

\draw  [color={rgb, 255:red, 74; green, 144; blue, 226 }  ,draw opacity=1, very thick] (186.5,129.5) .. controls (186.5,85.04) and (222.54,49) .. (267,49) .. controls (311.46,49) and (347.5,85.04) .. (347.5,129.5) .. controls (347.5,173.96) and (311.46,210) .. (267,210) .. controls (222.54,210) and (186.5,173.96) .. (186.5,129.5) -- cycle ;
\draw  [color={rgb, 255:red, 208; green, 2; blue, 27 }  ,draw opacity=1, very thick] (282,129.5) .. controls (282,93.88) and (310.88,65) .. (346.5,65) .. controls (382.12,65) and (411,93.88) .. (411,129.5) .. controls (411,165.12) and (382.12,194) .. (346.5,194) .. controls (310.88,194) and (282,165.12) .. (282,129.5) -- cycle ;
\draw[very thick]    (186.5,129.5) -- (347.5,129.5) ;

\draw (164.4,55.6) node [anchor=north west][inner sep=0.75pt]    {$\textcolor[rgb]{0.29,0.56,0.89}{\tilde{Z}}\textcolor[rgb]{0.29,0.56,0.89}{_{j-1}}$};
\draw (396.4,60.6) node [anchor=north west][inner sep=0.75pt]    {$\textcolor[rgb]{0.82,0.01,0.11}{x}\textcolor[rgb]{0.82,0.01,0.11}{_{j}}$};
\draw (252,158) node [anchor=north][inner sep=0.75pt]   [align=left] {$\mathrm{II}$};
\draw (249,84) node [anchor=north west][inner sep=0.75pt]   [align=left] {$\mathrm{I}$};
\draw (308,95) node [anchor=north west][inner sep=0.75pt]   [align=left] {$\mathrm{III}$};
\draw (308,147) node [anchor=north west][inner sep=0.75pt]   [align=left] {$\mathrm{IV}$};
\draw (367,122) node [anchor=north west][inner sep=0.75pt]   [align=left] {$\mathrm{V}$};
\end{tikzpicture}\par
}

\caption{Reference for sets involved in proof of Lemma~\ref{lem:parallelmismatchreduction}. The regions indicated are: $\mathrm{II}\cup \mathrm{IV} = \tilde{Z}'_{j-1}$,  $\mathrm{I}\cup \mathrm{III} = A_{j-1}$, $\mathrm{II} = \tilde{Z}'_{j}$, $\mathrm{IV} = x_j'$, $\mathrm{III}\cup \mathrm{IV} = \tilde{Z}_{j-1}\cap x_j$, $\mathrm{I}\cup \mathrm{V} = A_{j}$, and $\mathrm{I}\cup \mathrm{II}\cup \mathrm{V} = \tilde{Z}_{j}.$}~\label{fig:sets}
\end{figure}
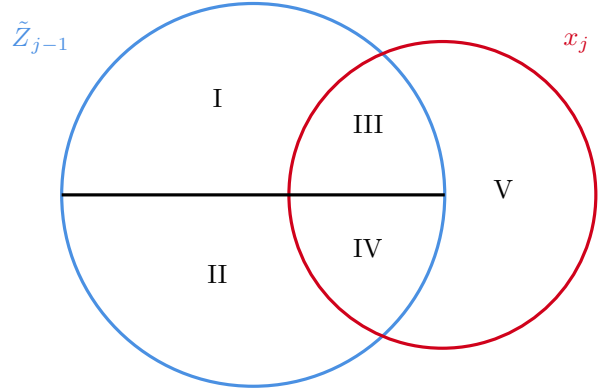

Referring to Figure~\ref{fig:sets}, we have $A_j\setminus A_{j-1} = x_j\setminus \tilde{Z}_{j-1}$, and hence
\begin{align}
    |A_j\setminus A_{j-1}| = |x_j\setminus \tilde{Z}_{j-1}| = |x_j| - |x_j \cap \tilde{Z}_{j-1}| \le |x_j| - \left(1-\frac{\varepsilon'}{2}\right)|x_j| = \frac{\varepsilon'}{2}|x_j|\,.\label{ineq:A_bound1}
\end{align}
We also have
\begin{align}
    (A_{j-1}\setminus A_j) \sqcup x_j' = \tilde{Z}_{j-1} \cap x_j\, ,
\end{align}
corresponding to the unions of regions $\mathrm{III}$ and $\mathrm{IV}$ in Figure~\ref{fig:sets}. If $j \in B$ is a ``bad'' index, then we have
\begin{align}
|A_{j-1}\setminus A_j|+\left(1-\frac{3}{2}\varepsilon'\right)|x_j| > |A_{j-1}\setminus A_j| + |x_j'| = |\tilde{Z}_{j-1} \cap x_j| \ge \left(1-\frac{\varepsilon'}{2}\right)|x_j|\,,
\end{align}
where the first inequality follows from the fact that $j\in B$ and the last from the decoding condition with parameter $\varepsilon'$. It follows that 
\begin{align}
    |A_{j-1}\setminus A_j| \ge \varepsilon'|x_j|,\label{ineq:A_bound2}
\end{align}
and hence
\begin{align}
    |A_{j-1}|-|A_j| = |A_{j-1}\setminus A_j| - |A_{j}\setminus A_{j-1}| \ge \varepsilon'|x_j| - \frac{\varepsilon'}{2}|x_j| = \frac{\varepsilon'}{2}|x_j|\,,
\end{align}
where we use inequalities~\eqref{ineq:A_bound1} and~\eqref{ineq:A_bound2} above. It follows that we have
\begin{equation}
\begin{cases}
    |A_j| - |A_{j-1}| \le \ \ \varepsilon'|x_j|/2, & \forall j\in G\,, \\
    |A_j| - |A_{j-1}| \le -\varepsilon'|x_j|/2, &\forall j \in B\,.
\end{cases}
\end{equation}
Summing the inequalities above, we get
\begin{align}
    0 \le |A_T| - |A_0| = \sum_{j=1}^T\left(|A_j| - |A_{j-1}|\right) \le \frac{\varepsilon'}{2}\left(\sum_{j\in G}|x_j| - \sum_{j \in B}|x_j|\right)\,,
\end{align}
where $|A_0| = 0$ by definition. Therefore
\begin{align}
    \sum_{j \in B}|x_j| \le \sum_{j\in G}|x_j|\,.
\end{align}
We have
\begin{align}
\sum_{j \in B}|x_j'| \le \left(1-\frac{3}{2}\varepsilon'\right)\sum_{j\in B}|x_j| \le \left(1-\frac{3}{2}\varepsilon'\right)\sum_{j\in G}|x_j| \le \sum_{j\in G}|x_j'|\,,
\end{align}
and hence
\begin{align}
|\tilde{Z}| - |\tilde{Z}_T| \le  \left|\bigcup_{j=1}^Tx_j'\right|= \sum_{j=1}^T|x_j'| = \sum_{j\in B}|x_j'| + \sum_{j\in G}|x_j'| \le 2\sum_{j\in G}|x_j'|\,.\label{ineq:claim1}
\end{align}
Now, consider the iteration of parallel decoding beginning with input $\tilde{Z}\equiv \tilde{Z}^{(k)}$. Let $u \in \mathbb{F}_2^Q$ denote the set of all qubits which have been acted on by the parallel decoder, i.e.,
\begin{align}
u = \bigcup_{z_v\in \mathcal{F}}z_v\,,
\end{align}
where $\mathcal{F} = \{z_v\}$ is the collection of all local codewords found by the decoder in the current iteration. We now prove that for all $j \in G$, we have $|x_j \cap u| \ge c|x_j|$ for some constant $c > 0$.

Fix some $x_j$ and let $v$ denote its anchoring vertex. Let us write $y = |x_j'\cap u|$. First, let us show that we must have 
\begin{align}
    |x_j'\setminus u| < \frac{3}{4}|x_j|\,.
\end{align}
Suppose otherwise. Then let $z_v$ denote the codeword (possibly zero) that the parallel decoder assigns to vertex $v$. Note that we have $z_v \subseteq u$ by definition, as well as
\begin{align}
    |\overline{Z}| - |\overline{Z} + z_v| \ge \frac{1}{2}|z_v|\,,\label{ineq:parallel_decoder_condition}
\end{align}
where $\overline{Z}$ denotes the current state of the noisy mismatch in the parallel decoder. By definition of $u$ as the execution support of the decoder, the qubits of $x_j'\setminus u$ are untouched by the algorithm. Therefore, since $x_j' \subseteq \tilde{Z}$, it follows that $x_j'\setminus u \subseteq \overline{Z}$ and $x_j'\setminus u \subseteq \overline{Z} + z_v$. Therefore we have
\begin{align}
    x_j' \setminus u = x_j'\setminus u \cap (\overline{Z} + z_v) \subseteq x_j \cap (\overline{Z} + z_v)\,.
\end{align}
The addition of $x_j$ to $\overline{Z}+z_v$ therefore removes at least $|x_j'\setminus u| \ge \frac{3}{4}|x_j|$ qubits from $\overline{Z}$. Consequently, the addition of $x_j$ to $\overline{Z}+z_v$ can add at most $|x_j|/4$ qubits, so that
\begin{align}
    |\overline{Z}+z_v| - |\overline{Z}+z_v+x_j| \ge \frac{1}{2}|x_j|\,.
\end{align}
Adding this inequality to~\eqref{ineq:parallel_decoder_condition}, we get
\begin{align}
    |\overline{Z}| - |\overline{Z} + z_v + x_j| \ge \frac{1}{2}(|x_j|+|z_v|) \ge \frac{1}{2}|z_v + x_j|\,.
\end{align}
Similar to the argument above, the addition of $x_j$ to $z_v$ adds at least $|x_j\setminus z_v| \ge |x_j'\setminus u| \ge 3|x_j|/4$ qubits, and hence removes at most $|x_j|/4$ qubits. Therefore
\begin{align}
    |z_v + x_j| - |z_v| \ge \frac{1}{2}|x_j|\,.
\end{align}
Since $|z_v+x_j| > |z_v|$, this contradicts the assumption that $z_v$ is the local codeword selected by the decoder, since the decoder will choose to maximize the Hamming weight of its local codewords. It follows that we've established the inequality
\begin{align}
    |x_j'\setminus u|  < \frac{3}{4}|x_j|\,.
\end{align}
This then implies that for all $j\in G$, we have
\begin{align}
    |x_j'\cap u|  >  |x_j'| - \frac{3}{4}|x_j| \ge  \left(1-\frac{3}{2}\varepsilon'\right)|x_j| - \frac{3}{4}|x_j| = \left(\frac{1}{4} - \frac{3}{2}\varepsilon'\right)|x_j|\,.
\end{align}
Since the $x_j'$ are disjoint, we get
\begin{align}
    |u| \ge \sum_{j \in G}|x_j' \cap u| > \left(\frac{1}{4}-\frac{3}{2}\varepsilon'\right)\sum_{j\in G}|x_j| \ge \left(\frac{1}{4}-\frac{3}{2}\varepsilon'\right)\sum_{j\in G}|x_j'| \ge \frac{1}{8}\left(1 - 6\varepsilon'\right)\left(|\tilde{Z}| - |\tilde{Z}_T|\right)\,,
\end{align}
where the last inequality follows from~\eqref{ineq:claim1}. Finally, by the decoding criterion~\eqref{ineq:parallel_decoder_condition}, the total decrease in mismatch weight is
\begin{align}
    |\tilde{Z}^{(k)}| - |\tilde{Z}^{(k+1)}| \ge \frac{1}{2}\sum_{z_v \in \mathcal{F}}|z_v| \ge \frac{1}{2}|u| \ge \frac{1}{16}(1-6\varepsilon')\left(|\tilde{Z}^{(k)}| - |\tilde{Z}_T^{(k)}|\right),
\end{align}
where we restore the superscript $(k)$ in this last inequality for clarity.
\end{proof}

Now, as in the sequential case, we bound the weight of the residual mismatch by the weight of measurement noise. 
\begin{lemma}\label{lem:parallelresidualmismatchbound}
Let $e$ be an error and $D$ be a syndrome noise. Let $\tilde{Z}$ be the initial mismatch vector assigned to $e$ and $D$. Let $\tilde{Z}^{(k)}$ denote the state of the mismatch vector after $k$ iterations of parallel decoding. Let $\tilde{Z}^{(k)}_{T}$ denote the residual mismatch vector obtained by running the sequential decoder with input $\tilde{Z}^{(k)}$ and parameter $\varepsilon'$.

Suppose that $A|e|_R + B|D|_V \le C_\delta n$. Then for all $k \in \mathbb{N}^+$ we have
\begin{align}
    |\tilde{Z}_T^{(k)}| \le \left(1+\frac{2(1-\delta)}{\varepsilon' - 2\delta}\right)\Delta^2|D|_V\equiv (1+\zeta)\Delta^2|D|_V\, .\label{eq:beta_bound}
\end{align}
\end{lemma}
\begin{proof}
Suppose that $\mathcal{F} = \{x_i\}_{i=1}^K$ are the codewords which have been found by the parallel decoder after $k$ iterations. Note that we can equivalently consider the same sequence to be obtained by running the sequential decoder with parameter $1/2$, i.e., we can consider $\tilde{Z}^{(k)}$ to be a state of the mismatch after $K$ iterations of sequential decoding with parameter $1/2$. It follows that $\tilde{Z}_{T}^{(k)}$ is a mismatch obtained by first running the sequential decoder with input $\tilde{Z}$ and parameter $1/2$ for $K$ iterations, and then switching to parameter $\varepsilon'$ for the remaining iterations. 

Applying Lemma~\ref{lemma:LZdecoderworks} with $\varepsilon = 1/2$, our assumptions on $|e|_R$ and $|D|_V$ imply that $Z^{(k)}_T$ is $\delta$-decomposable. Next, applying Lemma~\ref{lem:flipexists} (with $\varepsilon'$ as $\varepsilon$), it follows that
\begin{align}
    |Z_T^{(k)}| \le \frac{2(1-\delta)}{\varepsilon' - 2\delta}|Z_N|\,.
\end{align}
We then have
\begin{align}
    |\tilde{Z}_T^{(k)}| = |Z_T^{(k)} + Z_N| \le |Z_T^{(k)}| + |Z_N| \le \left(1 + \frac{2(1-\delta)}{\varepsilon' - 2\delta}\right)|Z_N| \le \left(1 + \frac{2(1-\delta)}{\varepsilon' - 2\delta}\right)\Delta^2|D|_V.
\end{align}
\end{proof}

For simplicity, we take $\varepsilon' = 3\delta$ in the following theorem. Note that this sets an upper bound on $\delta$ so that $\delta < 1/18$.

\begin{theorem}[Main Theorem for the Parallel Decoder]\label{thm:parallel_residual}
    Let $e$ be an error and $D$ be a syndrome error. Let $\tilde{Z}$ be the initial (noisy) mismatch associated with $e$ and $D$. Assume that
    \begin{align}
        A|e|_R+B|D|_V &\le \min\left(C_\delta n, d/\Delta\right)\, .
    \end{align}
    Then after $k$ iterations of parallel decoding, the decoder returns a correction $\hat{f}^{(k)}$ such that
    \begin{align}
    |e+\hat{f}^{(k)}|_R \le  \alpha_k |e|_R +  \beta |D|_V\,,
    \end{align}
    where
    \begin{align}
        \alpha_k = \frac{24}{5\kappa}(1-\gamma)^k,\quad \beta = \frac{6}{\kappa\delta}\Delta^2,\quad\text{and}\quad \gamma = (1-18\delta)/16\, .
    \end{align}
\end{theorem}
\begin{proof}
    Applying Lemmas~\ref{lem:parallelmismatchreduction} and~\ref{lem:parallelresidualmismatchbound}, it follows that the mismatch after $k$ iterations of parallel decoding is bounded above as
    \begin{align}
         |\td Z^{(k)}| \leq \left(1-\gamma\right)|\td Z^{(k-1)}| + \gamma(1+\zeta)\Delta^2|D|_V \, .
    \end{align}
    Summing this inequality over $k$ gives
    \begin{align}
        |\td Z^{(k)}| &\leq (1-\gamma)^k|\td Z| + \gamma(1+\zeta)\Delta^2|D|_V\left(1+(1-\gamma)+(1-\gamma)^2+\ldots + (1-\gamma)^{k-1}\right)\\
        &\leq (1-\gamma)^k|\td Z| + (1+\zeta)\Delta^2|D|_V\, .\label{eq:parallel_weight_bound}
    \end{align}
    Next, let $\tilde{e}^{(k)}$ denote the state of the error after $k$ iterations of parallel decoding. Let $\tilde{e}^{(k)}_T$ denote the state of the error after $T$ additional iterations of sequential decoding with parameter $\varepsilon'$. Let us write
    \begin{align}
        \tilde{e}^{(k)}_T = \tilde{e}^{(k)}+\sum_{i=1}^Tf_i,
    \end{align}
    where $\{f_i\}_{i=1}^T$ are the associated flip-sets with parameter $\varepsilon'$. It follows from Lemma~\ref{lemma:V10weightedmismatchinduction} that ${e}^{(k)}_T$ is $V_{10}$-weighted with associated mismatch ${Z}^{(k)}_T$. Lemma~\ref{lem:Z_Tsoundness} then implies that
    \begin{align}
        |{e}^{(k)}_T|_R\le \frac{1}{(1-\delta)\kappa}|Z^{(k)}_T|\le \frac{\zeta}{(1-\delta)\kappa}|Z_N| \le \frac{\zeta}{(1-\delta)\kappa}\Delta^2|D|_V.\label{eq:line0}
    \end{align}
    It remains to bound the weight of $|\td e^{(k)}|_R$. We have
    \begin{align}
        |\tilde{e}^{(k)}|_R &\le |e^{(k)}|_R + \Delta^2|D|_V\label{eq:line1}\\
        &\le \left|e^{(k)}_T + \sum_{i = 1}^Tf_i \right|_R + \Delta^2|D|_V\\
        &\le |e^{(k)}_T|_R + \sum_{i = 1}^T|f_i| + \Delta^2|D|_V\\
        &\le \frac{\zeta}{(1-\delta)\kappa}\Delta^2|D|_V + \frac{1}{\kappa}\sum_{i=1}^T|x_i| + \Delta^2|D|_V\label{eq:line2}\\
        &\le \left(1+\frac{\zeta}{(1-\delta)\kappa}\right)\Delta^2|D|_V + \frac{1}{(1-\varepsilon')\kappa}\left(|\tilde{Z}^{(k)}| - |\tilde{Z}_T^{(k)}|\right)\label{eq:line3}\\
        &\le \left(1+\frac{\zeta}{(1-\delta)\kappa}\right)\Delta^2|D|_V + \frac{1}{(1-\varepsilon')\kappa}|\tilde{Z}^{(k)}|\\
        &\le \left(1+\frac{\zeta}{(1-\delta)\kappa}\right)\Delta^2|D|_V + \frac{1+\zeta}{(1-\varepsilon')\kappa}\Delta^2|D|_V + \frac{(1-\gamma)^k}{(1-\varepsilon')\kappa}|\tilde{Z}|\, .\label{eq:line4}
    \end{align}
In the above, the first inequality~\eqref{eq:line1} follows from~\eqref{ineq:8}. Inequality~\eqref{eq:line2} follows from~\eqref{eq:line0} and the $\kappa$-product-expansion of the local code. Inequality~\eqref{eq:line3} follows from the fact that each local codeword $x_i$ satisfies the decoding condition with parameter $\varepsilon'$. Finally, inequality~\eqref{eq:line4} follows from~\eqref{eq:parallel_weight_bound}.

Using the fact that $|\td Z| \le 4|e|_R + \Delta^2 |D|_V$, we can rewrite the inequality above in terms of $|e|_R$ and $|D|_V$  following the same steps used in \eqref{ineq:4} to \eqref{ineq:5}. This gives us
\begin{align}
        |\tilde{e}^{(k)}|_R \le \left(1+\frac{\zeta}{(1-\delta)\kappa} + \frac{1+\zeta}{(1-\varepsilon')\kappa} + \frac{(1-\gamma)^k}{(1-\varepsilon')\kappa}\right)\Delta^2|D|_V + \frac{4(1-\gamma)^k}{(1-\varepsilon')\kappa}|e|_R \, .
    \end{align}
Finally, setting $\varepsilon' = 3\delta$, and using the fact that $\kappa \le 1$~\cite{KPTwoSided}, we can relax the inequality above slightly to get $4/((1-\varepsilon')\kappa)\le 24/(5\kappa)$, as well as
\begin{align}
    1+\frac{\zeta}{(1-\delta)\kappa} + \frac{1+\zeta}{(1-\varepsilon')\kappa} + \frac{(1-\gamma)^k}{(1-\varepsilon')\kappa} &\le \frac{1}{\kappa}\left(1 + \frac{2}{\delta} + \frac{2-\delta}{1-3\delta}\cdot\frac{1}{\delta} + \frac{1}{1-3\delta}\right)\\
    &\le \frac{1}{\kappa}\left(1 + \frac{2}{\delta} + \frac{2}{1-3\delta}\cdot\frac{1}{\delta}\right)\\
    &\le \frac{6}{\kappa\delta}\, ,
\end{align}
which holds for $\delta \in (0, 1/18)$.
\end{proof}

\section{Discussion}\label{sec:discussion}
In our article, we have shown that quantum Tanner codes admit single-shot QEC. Given information from a single round of noisy measurements, the mismatch decomposition decoder~\cite{leverrier2022parallel} is able to output a correction that is close to the data error that occurred. For a variety of noise models, including adversarial or stochastic noise, the single-shot decoder is able to maintain the encoded quantum information for up to an exponential number of correction rounds. The parallelized version of the decoder can be run in constant time while keeping the residual error small. During readout, a logarithmic number of iterations suffices to recover the logical information. 

One may also ask about the possibility of single-shot QEC with other decoders for good QLDPC codes. Due to the close connection between the decoders analyzed here and the potential-based decoder for quantum Tanner codes in Ref.~\cite{gu2022efficient} (for example, the ability to map between candidate flip sets for both types of decoders), a corollary of the proofs presented here is that the potential-based decoder also has the single-shot property. Likewise, under the mapping of errors shown in Ref.~\cite{leverrier2022sequential}, the decoders considered here are applicable to the original good QLDPC codes by Panteleev and Kalachev~\cite{PK21}. Our analysis does not straightforwardly carry over to the code and decoder proposed in Ref.~\cite{dinur2022good}, and it remains to be seen whether that construction also admits single-shot decoding.

We further remark that all known constructions of asymptotically good QLDPC codes admit a property called small-set (co)boundary expansion~\cite{kaufman2013high}, which in the case of quantum Tanner codes, was used to prove the No Low-Energy Trivial States (NLTS) conjecture (see Property~1 of reference~\cite{anshu2022nlts}). 
Small-set (co)boundary expansion is also equivalent to the notion of soundness~\cite{PRXQuantum.2.020340}, which lower bounds the syndrome weight by some function of reduced error weight. Indeed, soundness is a strong indication of single-shot decodability. Similarly, quantum locally testable codes~\cite{aharonov2015quantum,hastings2016quantum,leverrier2022towards,cross2022quantum, wills2023general} admit analogous soundness properties, although decoders for such codes are unexplored. Note that in our proof, what we needed was a notion of soundness for the mismatch vector (see Lemma~\ref{lem:Zsoundness}), which is distinct from the usual notion of soundness for the syndrome. The weight of the mismatch is in general incomparable to the weight of the syndrome, so the precise relation between these two definitions of soundness is not well understood.

In conclusion, our results can be viewed as a step toward making general QLDPC codes more practical. While many challenges still remain, there have been promising developments in this direction~\cite{gottesman2014faulttolerant,Cohen2022,Tremblay2022,pattison2023hierarchical}. We believe that quantum LDPC codes, similar to classical LDPC codes, will constitute the gold standard for future quantum telecommunication technologies and form the backbone of resource-efficient quantum fault-tolerant protocols.

\begin{acknowledgements}
We would like to thank Robert K\"onig, Anthony Leverrier and Chris Pattison for inspiring discussions on single-shot QEC and QLDPC codes. S.G. acknowledges funding from the U.S. Department of Energy (DE-AC02-07CH11359), and the National Science Foundation (PHY-1733907). The Institute for Quantum Information and Matter is an NSF Physics Frontiers Center. E.T. acknowledges funding from the Sloan Foundation, DARPA 134371-5113608, and DOD KK2014. L.C. and S.C. gratefully acknowledge support by the European Research Council under grant agreement no.~101001976 (project EQUIPTNT). Z.H. would like to thank NSF grant CCF 1729369 for support.

\section*{Data availability}

Data sharing not applicable to this article as no datasets were generated or analyzed during the current study.
\end{acknowledgements}

\interlinepenalty=10000
\bibliography{bib}

\end{document}